\documentclass[preprintnumbers,floatfix,superscriptaddress,nofootinbib,prx,twocolumn,showpacs]{revtex4-2}

\usepackage[utf8]{inputenc}
\usepackage[english]{babel} 
\usepackage[T1]{fontenc} 

\usepackage{tabularx}
\usepackage{xspace}
\usepackage{dsfont}
\usepackage{pifont}
\usepackage{tikz}
\usepackage[hmargin=1.5cm, vmargin=2cm]{geometry}
\graphicspath{{./figs/}}
\usepackage{xcolor}

\usepackage{braket}
\usepackage{enumitem}
\usepackage{multirow}
\usepackage{changepage}
\usepackage{physics}
\usepackage{framed}
\usepackage{wasysym}
\usepackage{floatrow}

\usepackage{fancyhdr}
\usepackage{mathtools,amsmath,amsfonts,amsthm,stmaryrd,amssymb}
\numberwithin{equation}{section} 
\usepackage{empheq}
\usepackage{siunitx}

\newcommand{\defeq}{\vcentcolon=}

\usepackage{bm} 

\theoremstyle{definition}

\newtheorem{definition}{Definition}

\theoremstyle{plain}
\newtheorem{lemma}{Lemma}
\newtheorem{theorem}{Theorem}

\usepackage{caption}
\usepackage{booktabs}
\usepackage{graphicx}
\usepackage{subcaption}
\usepackage[font=footnotesize,labelfont=bf]{caption}
\usepackage{ragged2e}
\DeclareCaptionJustification{justified}{\justifying}
\usepackage{hyperref}
\hypersetup{
    colorlinks=true,
    linkcolor=blue,
    filecolor=magenta,      
    urlcolor=cyan,
    citecolor=magenta
}

\usepackage{float}
\floatstyle{plaintop}
\restylefloat{table}
\usepackage{array}
\newcolumntype{C}[1]{>{\centering\let\newline\\\arraybackslash\hspace{0pt}}m{#1}}

\usepackage{tikz}
\usetikzlibrary{automata,arrows,positioning,calc}

\usepackage{algorithm}
\usepackage{algpseudocode}
\makeatletter
\def\BState{\State\hskip-\ALG@thistlm}
\makeatother

\newcommand\tc{t_{\mathrm{cut}}}

\allowdisplaybreaks

\usepackage{microtype}       
\newcommand{\ssec}[1]{\textls{\textsc{\textbf{#1}}}}

\begin{document}

\title{Performance metrics for the continuous distribution of entanglement\\in multi-user quantum networks}
\date{\today}

\author{\'{A}lvaro G. I\~{n}esta}\email{a.gomezinesta@tudelft.nl}
\affiliation{QuTech, Delft University of Technology, Lorentzweg 1, 2628 CJ Delft, The Netherlands}
\affiliation{EEMCS, Quantum Computer Science, Delft University of Technology, Mekelweg 4, 2628 CD Delft, The Netherlands}
\affiliation{Kavli Institute of Nanoscience, Delft University of Technology, Lorentzweg 1, 2628 CJ Delft, The Netherlands}
\author{Stephanie Wehner}
\affiliation{QuTech, Delft University of Technology, Lorentzweg 1, 2628 CJ Delft, The Netherlands}
\affiliation{EEMCS, Quantum Computer Science, Delft University of Technology, Mekelweg 4, 2628 CD Delft, The Netherlands}
\affiliation{Kavli Institute of Nanoscience, Delft University of Technology, Lorentzweg 1, 2628 CJ Delft, The Netherlands}
\begin{abstract}
Entangled states shared among distant nodes are frequently used in quantum network applications.
When quantum resources are abundant, entangled states can be continuously distributed across the network, allowing nodes to consume them whenever necessary. This continuous distribution of entanglement enables quantum network applications to operate continuously while being regularly supplied with entangled states.
Here, we focus on the steady-state performance analysis of protocols for continuous distribution of entanglement.
We propose the virtual neighborhood size and the virtual node degree as performance metrics.
We utilize the concept of Pareto optimality to formulate a multi-objective optimization problem to maximize the performance.
As an example, we solve the problem for a quantum network with a tree topology.
One of the main conclusions from our analysis is that the entanglement consumption rate has a greater impact on the protocol performance than the fidelity requirements.
The metrics that we establish in this manuscript can be utilized to assess the feasibility of entanglement distribution protocols for large-scale quantum networks.
\end{abstract}

\maketitle

\section{Introduction}
Quantum networks are expected to enable multi-party applications that are provably impossible by using only classical information. These applications range from basic routines, such as quantum teleportation \cite{Bennett1993,Gottesman1999}, to more complex tasks, such as quantum key distribution \cite{Ekert1991,Bennett1992} and entanglement-assisted distributed sensing \cite{Xia2020,Grace2021}.
Some of these applications may operate in the \emph{background} (e.g., a quantum key distribution subroutine that is continuously generating secret key), as opposed to \emph{sporadic} applications that are executed after the users actively trigger them.
Most quantum network applications consume shared entanglement as a basic resource.
\emph{Entanglement distribution protocols} are used to generate and share multipartite entanglement among remote parties.
There are two main approaches to distribute entanglement among the nodes~\cite{Chakraborty2019, Illiano2022}:
\begin{itemize}
	\item Protocols for \textit{on-demand distribution of entanglement} distribute entangled states only after some nodes request them. The request may involve some quality-of-service requirements (e.g., a minimum quality of the entanglement). This type of protocol typically involves solving a routing problem and scheduling a set of operations on a subset of nodes~\cite{Briegel1998, Chakraborty2019, Victora2020, Vardoyan2021b, Inesta2023}.
	\item Protocols for \textit{continuous distribution of entanglement} (CD protocols) continuously distribute entangled states among the nodes. These entangled states can be consumed by the nodes whenever they need them. This allows \emph{background applications} to continuously operate and consume entanglement in the background.
	In this work, we focus on CD protocols that provide entanglement to background applications.
\end{itemize}

On-demand distribution is generally more efficient, since entanglement is only produced when it is needed.
This makes on-demand distribution more suitable for quantum networks where the quantum resources are limited (e.g., networks with a small number of qubits per node).
As a consequence, previous work, both theoretical \cite{Briegel1998, Chakraborty2019, Victora2020, Vardoyan2021b, Inesta2023} and experimental \cite{Ursin2007, Moehring2007, Slodicka2013, Bernien2013, Hensen2015, Humphreys2018, Pompili2021}, has mostly focused on this type of protocol in quantum networks with a simple topology or with very limited number of qubits per node.

On-demand distribution requires a scheduling policy that tells the nodes when to perform each operation based on specific demands.
If the number of nodes involved in the generation of entanglement is large, the scheduling policies become more complex.
In contrast, the continuous distribution of entanglement does not necessarily require an elaborate application-dependent schedule.
Therefore, CD protocols are expected to allocate resources faster and prevent traffic congestion in large quantum networks.
Here, we focus on the performance evaluation of CD protocols.
Specifically, we consider protocols that distribute bipartite entanglement among remote nodes. We refer to shared bipartite entanglement as an \emph{entangled link}.
We focus on entangled links because this is a basic resource needed in many quantum network applications \cite{Ekert1991,Bennett1992,BenOr2006,Leichtle2021}, where nodes generally need many copies of a bipartite entangled state with high enough quality.
Even when multipartite entanglement is required, it can be generated using entangled links~\cite{Kruszynska2006,Pirker2018,Meignant2019,Bugalho2021}.

We consider a quantum network with $n$ nodes. Some pairs of nodes are \emph{physical neighbors}: they are connected by a physical channel, such as optical fibers \cite{Yoshino2013,Stephenson2020} or free space \cite{Ursin2007,Sidhu2021}.
This is depicted in Figure~\ref{fig.quantum-network}.
To generate long-distance entanglement, we assume the nodes can perform the following basic operations:
($i$) heralded generation of entanglement between physical neighbors \cite{Barrett2005,Bernien2013}, which successfully produces an entangled link with probability $p_\mathrm{gen}$ and otherwise raises a failure flag;
($ii$) entanglement swaps \cite{Zukowski1993, Duan2001, Sangouard2011}, which consume one entangled link between nodes A and B and another entangled link between nodes B and C to generate a single link between A and C with probability $p_\mathrm{s}$;
($iii$) removal of any entangled link that has existed for longer than some cutoff time $\tc$ to prevent the existence of low-quality entanglement in the network \cite{Collins2007, Rozpedek2018, Khatri2019, Rozpedek2019, Li2020};
and ($iv$) consumption of entangled links in background applications at some constant rate $p_\mathrm{cons}$.
Note that the choice of cutoff time is determined by the minimum fidelity required by the applications, $F_\mathrm{app}$.
We allow for multiple entangled links to be shared simultaneously between the same pair of nodes (see Figure~\ref{fig.quantum-network}).

\begin{figure}[t]
\captionsetup[subfigure]{justification=centering}
\centering
  \centering
  \includegraphics[width=\linewidth]{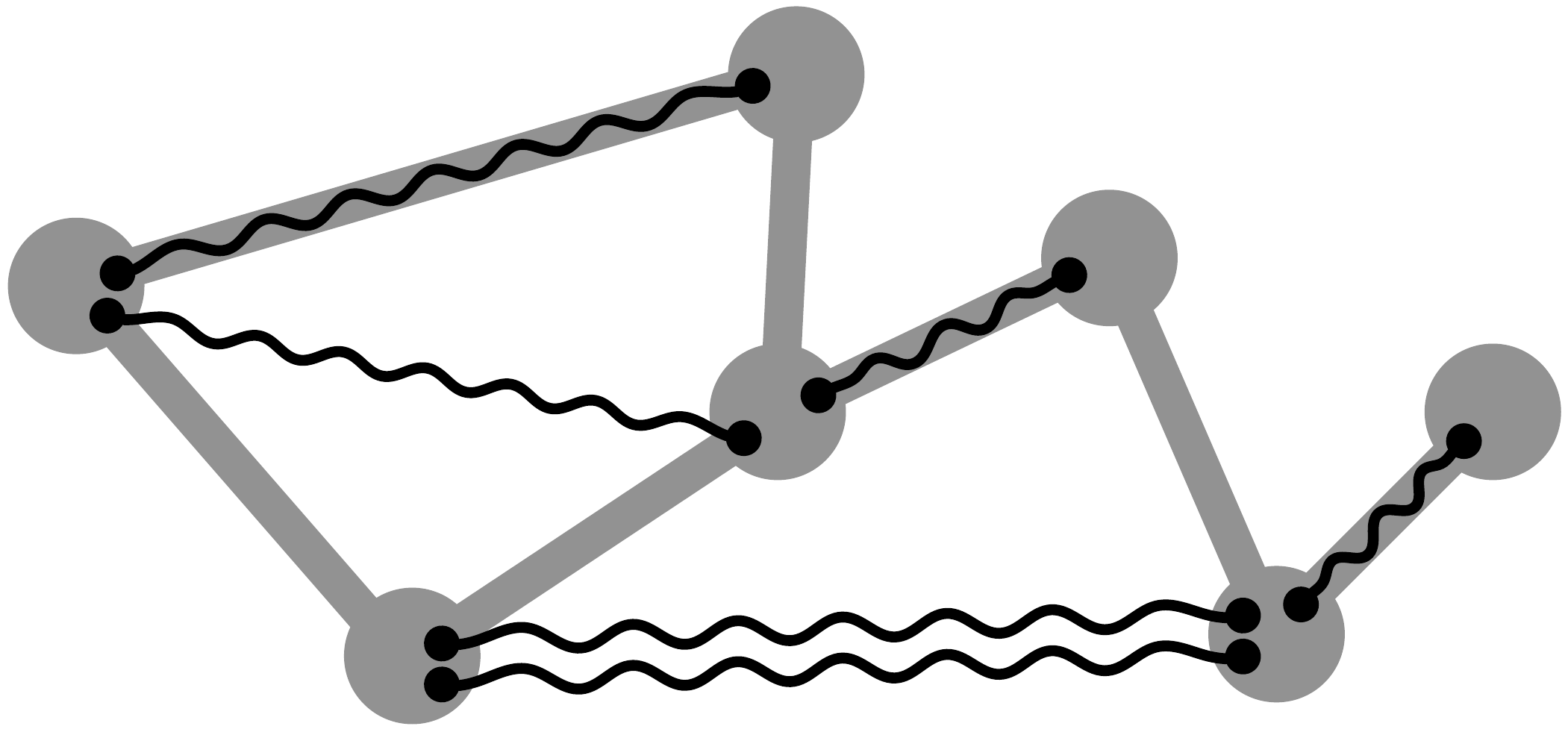}
\caption{\textbf{Illustration of a seven-node quantum network.}
The nodes are represented as gray circles, and physical channels connecting neighboring nodes are represented as gray lines.
Entangled links are represented as black lines connecting two occupied qubits (small black circles).
The physical topology is static, while the entangled links are continuously created, discarded, and consumed.}
\label{fig.quantum-network}
\end{figure}

Evaluating the performance of a CD protocol is a fundamentally different problem to evaluating the performance of on-demand protocols, since each type of protocol serves a different purpose.
In on-demand protocols, one generally wants to maximize the rate of entanglement distribution among a specific set of end nodes and the quality of the entanglement (or some combined metric, such as the secret key rate \cite{Gottesman2004}).
By contrast, the goal of a CD protocol is ($i$) to distribute entanglement among the nodes such that it can be continuously consumed in background applications and ($ii$) to ensure that some entanglement is available for sporadic applications.
To quantify the performance of a CD protocol, we need metrics that take these goals into account.
A simple approach is to analyze the configuration of entangled links that a CD protocol can achieve.
This configuration is time-dependent due to the dynamic nature of the entangled links.
Most previous work aimed at describing the connectivity of large-scale quantum networks disregards the time-dependence of the system.
As a consequence, previous results do not depend explicitly on parameters that determine the evolution of the entanglement, such as the coherence time.
For example, in Refs. \cite{Brito2020} and \cite{Brito2020a}, the authors study a graph in which the edges are entangled links that exist at a specific instant.
Some authors have described the connectivity of a quantum network using percolation theory \cite{Acin2007, Cuquet2009, Perseguers2010, Cuquet2011, Wu2011, Choi2019, Meng2021}, which also disregards the time-evolution of the entangled states, and often assumes specific topologies and some form of pre-shared entanglement.
Another line of related work is the use of pre-shared entanglement for on-demand applications \cite{Pouryousef2022, Kolar2022}.

In this paper, we consider quantum networks with arbitrary topologies where entanglement is continuously being generated and consumed.
We propose metrics to evaluate the performance of CD protocols. These metrics take into account the time-dependence of the system and can be used to optimize the protocol performance.
\\

Our main contributions are the following:
\begin{itemize}
	\item We define metrics to evaluate the performance of CD protocols in heterogeneous quantum networks with an arbitrary topology, namely, we define the \emph{virtual neighborhood size} and the \emph{virtual node degree}. These metrics provide information about the number of nodes that are able to continuously run background applications and about the number of existing backup entangled links to run sporadic applications.
	\item We provide analytical and numerical tools to compute the performance metrics.
	\item We provide a mathematical framework to maximize the virtual neighborhood size of every node in a heterogeneous network, while providing some minimum quality-of-service requirements (e.g., a minimum number of backup links). We do this via the concept of Pareto optimality.
	\item We study the relation between the steady-state performance of the entanglement distribution protocol and the application requirements (minimum fidelity and link consumption rate) in a quantum network with a tree topology.
\end{itemize}

Our main findings are the following:
\begin{itemize}
	\item The expected virtual neighborhood size rapidly drops to zero when the entanglement consumption rate increases beyond the entanglement generation rate.
	\item In a quantum network with a tree topology and with high entanglement generation rate,
	the consumption rate has a stronger effect on the virtual neighborhood size than the minimum fidelity required by the applications.
	In other words, background applications that require a high consumption rate affect the CD protocol performance more than applications that require a high fidelity.
	\item The set of protocol parameters that maximize the virtual neighborhood size is node-dependent. Consequently, in heterogeneous networks with an arbitrary topology we need to solve a multi-objective optimization problem.
\end{itemize}

The structure of the paper is as follows.
In Section \ref{sec.model}, we define the network model (physical topology, quantum operations, and quantum resources).
In Section \ref{sec.cdprotocol}, we provide an example of a CD protocol.
In Section \ref{sec.virtual}, we formally define the virtual neighborhood and the virtual node degree.
We apply these definitions to evaluate the performance of a CD protocol using analytical and numerical methods.
As an example, we analyze a CD protocol in a quantum network with a tree topology.
In Section \ref{sec.discussion}, we discuss the implications and limitations of our work.

\section{Network model}\label{sec.model}
In this Section we describe the physical topology of the network and the quantum operations that the nodes can perform. We also discuss the background applications requirements and the management of quantum resources at each node.

We consider a quantum network with $n$ \emph{nodes} (see Figure \ref{fig.quantum-network}).
Nodes can store quantum states in the form of qubits, and they can manipulate them as we describe below.
Additionally, some nodes are connected by a \emph{physical channel} over which they can send quantum states.
Qubits can be realized with different technologies, such as nitrogen vacancy (NV) centers \cite{Bernien2013, Hensen2015, Humphreys2018, Rozpedek2019, Pompili2021}, trapped ions \cite{Moehring2007, Slodicka2013}, or neutral atoms \cite{Welte2018}, while physical channels can be realized with optical fibers \cite{Yoshino2013,Stephenson2020} or free space \cite{Ursin2007,Sidhu2021}.
\\

\ssec{Physical topology.}
Two nodes are \emph{physical neighbors} if they share a physical channel.
The \emph{physical node degree} $d_i$ of node $i$ is the number of its physical neighbors.
The set of nodes and physical channels constitute the \emph{physical topology} of the quantum network.
Early quantum networks are expected to have simple physical topologies, such as a chain where each node is connected to two other nodes \cite{Briegel1998, Coopmans2021, Inesta2023} and a star topology where all nodes are only connected to a central node \cite{Vardoyan2021,Vardoyan2021b}.
More advanced networks are expected to display a more complex physical topology, such as a dumbbell structure with a backbone connecting two metropolitan areas.

The definitions and methods we develop in this work are general and apply to an arbitrary physical topology, which can be described using an \emph{adjacency matrix} $A$ (element $A_{ij}$ is 1 if nodes $i$ and $j$ are physical neighbors and 0 otherwise).
To illustrate how our methods can be valuable and effective, we apply them to a quantum network with a tree topology as an example.
In a tree, any node can be reached from any other node by following exactly one path.
This topology is particularly relevant as it has been shown that it requires a reduced number of qubits per node to avoid traffic congestion~\cite{Choi2023}.
\begin{definition}
	A \emph{($d$,$k$)-tree} network is an undirected unweighted graph where nodes are distributed in $k$ levels, with $d^l$ nodes in level $l\in{0,1,\dots,k-1}$. Each node in level $l$ is connected to $d$ nodes in the ($l+1$)-th level, and is only connected to one node in the ($l-1$)-th level.
\end{definition}
The total number of nodes in a ($d$,$k$)-tree is $n=(d^k-1)/(d-1)$, and the network diameter is $2k$.
A (2,3)-tree network is depicted in Figure \ref{fig.23tree}.
\\

\ssec{Entanglement distribution.}
The aim of a CD protocol is to distribute shared bipartite entangled states, which we call entangled links.
Ideally, entangled links are maximally entangled states. However, entanglement generation and storage are generally noisy processes. Consequently, we assume that entangled links are Werner states \cite{Werner1989}: maximally entangled states that have been subjected to a depolarizing process, which is a worst-case noise model \cite{Dur2005}.
Werner states can be written as
\begin{equation}
    \rho = \frac{4F-1}{3}\ketbra{\phi^+} + \frac{1-F}{3} \mathbb{I}_4,
\end{equation}
where $\ket{\phi^+}=(\ket{00}+\ket{11})/\sqrt{2}$ is a maximally entangled state, $F$ is the fidelity of the Werner state to the state $\ket{\phi^+}$, and $\mathbb{I}_m$ is the $m$-dimensional identity.
Here, the fidelity of a mixed state $\rho$ to a pure state $\ket{\phi}$ is defined as
\begin{equation}
    F(\rho,\ket{\phi}) \defeq \bra{\phi}\rho\ket{\phi}.
\end{equation}

\begin{figure}[t]
\captionsetup[subfigure]{justification=centering}
\centering
  \centering
  \includegraphics[width=0.6\linewidth]{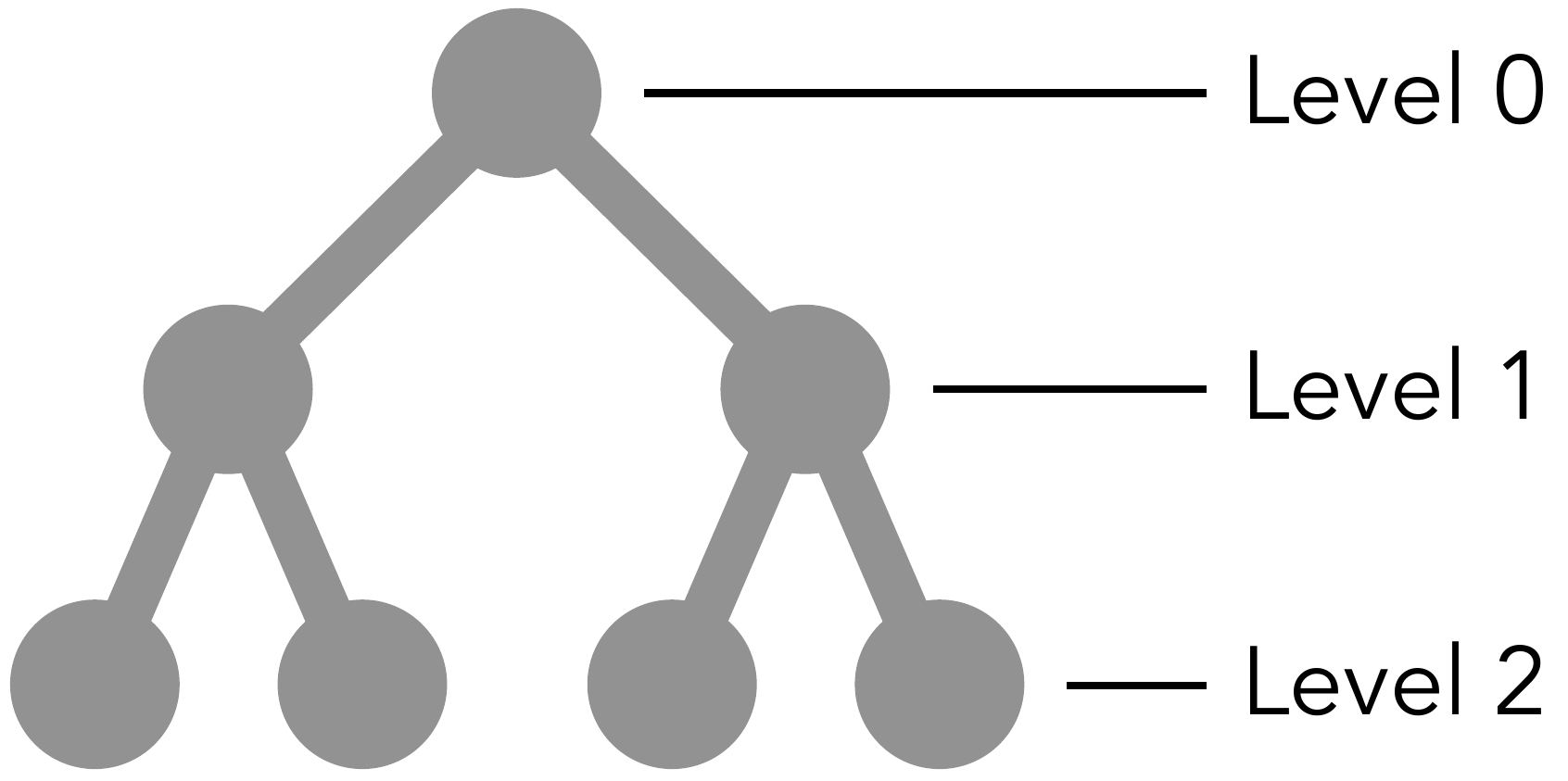}
\caption{\textbf{(2,3)-tree network.} Each node is represented as a gray circle and is connected to two other nodes in a lower level.}
\label{fig.23tree}
\end{figure}

We consider nodes that operate as first or second generation quantum repeaters \cite{Muralidharan2016}: physical neighbors generate entangled links via \emph{heralded entanglement generation} using two-way signaling.
This operation produces an entangled link with probability $p_\mathrm{gen}$ and otherwise raises a failure flag \cite{Barrett2005,Bernien2013}.
The fidelity of newly generated links, $F_\mathrm{new}$ is generally a function of $p_\mathrm{gen}$.
For example, in the single-photon protocol \cite{Hermans2023}, $F_\mathrm{new} = 1 - \lambda p_\mathrm{gen}$, for some $0 \leq \lambda \leq 1$ (as discussed in ref. \cite{Davies2023}, the value of $\lambda$ can be tuned by performing a batch of entanglement attempts as a single entanglement generation step \cite{Pompili2021a}).

Long-distance entanglement between physically non-neighboring nodes can be generated using \emph{entanglement swapping} \cite{Zukowski1993, Duan2001, Sangouard2011}, which consumes an entangled link between nodes A and B, with fidelity $F_\mathrm{AB}$, and another one between B and C, with fidelity $F_\mathrm{BC}$, to produce a link between A and C, with fidelity $F_\mathrm{AC} \leq F_\mathrm{AB}, F_\mathrm{BC}$. This operation succeeds with probability $p_\mathrm{swap}$ (when it fails, both input links are lost and nothing is produced).
Note that entanglement swapping also requires two-way classical signaling. See Appendix \ref{app.model} for further details on entanglement swapping.
\\

\ssec{Quantum applications.}
The main goal of a CD protocol is to provide a continuous supply of entanglement for nodes to run applications without the need for explicitly demanding entanglement.
We assume that each pair of nodes that share entanglement is continuously running quantum applications in the background, consuming entangled links at a rate $p_\mathrm{cons}$.
For simplicity, we assume $0\leq p_\mathrm{cons} \leq 1$.
Since we will assume time to be slotted (see Section \ref{sec.cdprotocol}), a consumption rate between zero and one can be interpreted as the probability that, in each time slot, two nodes that share some entangled links consume one link.
We consider entanglement purification \cite{Dur2007, Hartmann2007, Victora2020} as an application and therefore we omit it in our model (purification at the physical link level can be included in our model by modifying $p_\mathrm{gen}$ and $F_\mathrm{new}$ accordingly; see Appendix \ref{app.model} for further details).

Background applications require entanglement of a high enough quality. Specifically, we assume that they need entangled links with fidelity larger than $F_\mathrm{app}$.
\\

\ssec{Mitigating decoherence.}
The operations involving entangled links and the storage in memory have a negative impact on the quality of the links.
Each entanglement swap produces a link with a lower fidelity than the input links \cite{Munro2015}.
To prevent the fidelity from dropping too low, we must limit the \emph{maximum swap distance}, defined as the maximum number of short-distance links that can be combined into longer distance entanglement via swaps. We denote this maximum number of links as $M$. Two nodes can only share entanglement if they are at most $M$ physical links away.

Additionally, the fidelity of entangled links stored in memory decreases over time due to couplings to the environment \cite{Dur2005,Chirolli2008}, making old links unusable for applications that require high fidelity states.
A simple technique to alleviate the effects of noisy storage consists in imposing a cutoff time $t_\mathrm{cut}$: any link that has been stored for longer than the cutoff time must be discarded \cite{Rozpedek2018}.

To ensure that the fidelity of every entangled link is above $F_\mathrm{app}$ in a network where new links are generated with fidelity $F_\mathrm{new}$, it is enough to choose the values of $t_\mathrm{cut}$ and $M$ such~that~\cite{Inesta2023}
\begin{equation}\label{eq.cutoffcondition}
    t_\mathrm{cut} \leq -T \ln\Bigg(\frac{3}{4F_\text{new}-1} \Big( \frac{4F_\text{app}-1}{3} \Big)^{\frac{1}{M}} \Bigg),
\end{equation}
where $T$ is a parameter that characterizes the exponential decay in fidelity of the whole entangled state due to the qubits being stored in noisy memories.
In our analysis, we choose the largest cutoff that satisfies (\ref{eq.cutoffcondition}).
For further details on the noise model, see Appendix \ref{app.model}.
\\

\ssec{Limited quantum resources.}
Nodes have a limited number of qubits.
These qubits can be used for communication (short coherence times) or for storage (long coherence times) \cite{Benjamin2006,Lee2022}.
Here, we assume a simplified setup where every qubit can be used for entanglement generation and for storage of an entangled link.
Intuitively, nodes with a larger number of physical neighbors should have more resources available, to establish entanglement with many neighbors simultaneously.
We assume that the maximum number of qubits that node $i$ can store is $d_i r$, where $d_i$ is the physical node degree of node $i$ and $r\in\mathbb{N}$ is a hardware-dependent parameter that limits the maximum number of qubits per node.

We make an additional simplifying assumption: each qubit can only generate entanglement with a fixed neighboring node. The physical motivation behind this assumption is the lack of optical switches in the node.
This assumption allows us to uniquely identify each qubit using a three-tuple address $(i,j,m)$. The first index, $i\in\{0,...,n-1\}$, corresponds to the node holding the qubit. The second index, $j\in\{0,...,n-1\}$, is the node with which the qubit can generate entanglement ($i\neq j$). The third index, $m\in\{0,...,r-1\}$, is used to distinguish qubits that share the same indices $i$ and $j$. A graphical example is shown in Figure~\ref{fig.qubitaddresses}.
\\

\begin{figure}[t]
\captionsetup[subfigure]{justification=centering}
\centering
  \centering
  \includegraphics[width=0.7\linewidth]{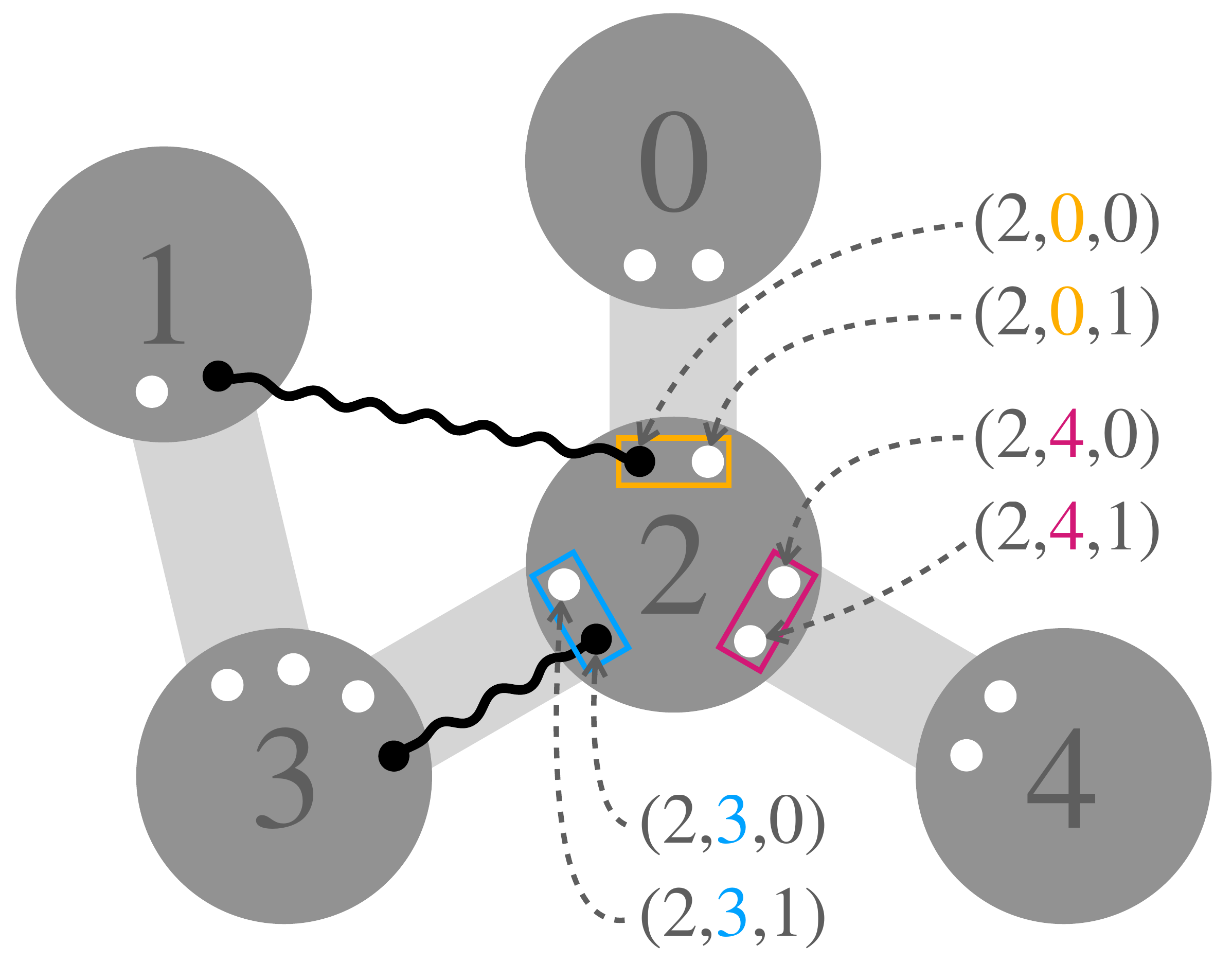}
\caption{\textbf{Qubit addresses.} Each qubit is identified by a qubit address consisting of three values $(i,j,m)$: $i$ is the node holding the qubit, $j$ is the neighboring node that can generate entanglement with that qubit, and $m$ is used to distinguish qubits with the first two indices $i$ and $j$. In this example, each node has two qubits per physical neighbor, i.e., $r=2$.}
\label{fig.qubitaddresses}
\end{figure}

\section{Protocol for continuous distribution of entanglement}\label{sec.cdprotocol}
The operations discussed above --entanglement generation, swaps, entanglement consumption, and application of cutoffs-- are performed following a specific protocol for continuous distribution of entanglement (CD protocol).
Here we consider a basic CD protocol that we will use to test our performance optimization tools.
We assume a synchronous protocol: time is divided into non-overlapping time slots and each operation is allocated within a time slot.
This is a common assumption in the field of quantum networking (see, e.g., Refs. \cite{Skrzypczyk2021, Inesta2023}), since nodes generally have to agree to perform synchronized actions for heralded entanglement generation.
In what follows, we focus on the Single Random Swap (SRS) protocol, which is described in Algorithm \ref{alg.srs}.
In this protocol, ($i$) entanglement generation is attempted sequentially on every physical link; ($ii$) swaps are performed using links chosen at random; and ($iii$) every pair of nodes that shares an entangled link consumes one link per time step with probability $p_\mathrm{cons}$.
The protocol has a single parameter, $q\in[0,1]$, which determines how many nodes must perform a swap at each time step (if $q=0$, no swaps are performed; if $q=1$, every node must perform a swap if possible; if $0<q<1$, a random subset of nodes may perform swaps).
In step 3.2 of the SRS protocol, the condition $A_{jk}=0$, ensures that swaps will generally not connect physical neighbors.

In step 5, we remove links that have too low fidelity since they were produced after swapping too many shorter links. To stop these links from forming in the first place, we would need to consider a more complex swapping policy where nodes are allowed to coordinate their actions (or a simple policy where communication is assumed to be instantaneous).

In Table \ref{tab.variables} we provide a summary of the network and protocol parameters.
In the next Section we present our performance metrics and how to use them to tune the protocol parameter(s) for an optimal performance.
Note that our methods can be applied to any other (synchronous and non-synchronous) CD protocol.

\begin{algorithm}
\caption{- SRS entanglement generation protocol.}\label{alg.srs}
\begin{flushleft}
	\vspace{5pt}
	\textbf{Inputs:} \\
	\begin{itemize}[label=-]
		\item Quantum network with an arbitrary configuration of entangled links and
		\begin{itemize}[label=·]
			\item physical adjacency matrix  $A$;
			\item probability of successful entanglement generation $p_\mathrm{gen}$;
			\item probability of successful swap $p_\mathrm{s}$;
			\item maximum swap distance $M$;
			\item probability of link consumption $p_\mathrm{cons}$.
		\end{itemize}
		\item $q$: probability of performing a swap.
	\end{itemize}
	\vspace{5pt}
	\textbf{Outputs:}
	\begin{itemize}[label=-]
		\item Quantum network with updated configuration of links.
	\end{itemize}
	\vspace{5pt}
	\textbf{Algorithm:}
\end{flushleft}
\begin{algorithmic}[1]
\State \textbf{Cutoffs} are applied and old links are removed.
\State \textbf{Entanglement generation} is attempted at every physical link if enough qubits are available. One entangled link is generated at each physical link with probability $p_\mathrm{gen}$.
\State \textbf{Swaps} are performed. Every node $i$ does the following, in parallel to each other:
	\Statex \hspace{2pt} 3.1: Pick at random a qubit entangled to some qubit in another node $j$.
	\Statex \hspace{2pt} 3.2: Pick at random a qubit entangled to some qubit in node $k\neq j$, and with $A_{jk}=0$. If not possible, go to step 4.
	\Statex \hspace{2pt} 3.3: With probability $q$, perform a swap on both qubits, which succeeds with probability $p_\mathrm{s}$. If it fails, both links involved in the swap are discarded.
\State \textbf{Classical communication}: every node gains updated information about every qubit (where it is connected to) and about every entangled link (link age and number of swaps used to create the link).
\State \textbf{Long links removal}: links that were produced as a consequence of swapping more than $M$ elementary-level links are removed.
\State \textbf{Consumption}: each pair of nodes that share links consume one of them with probability $p_\mathrm{cons}$.
\end{algorithmic}
\end{algorithm}

\renewcommand{\arraystretch}{1.2}
\begin{table}[t]
    \centering
    \caption{Parameters of the quantum network.
    The number of nodes is given by the size of the adjacency matrix $A$.
    When considering a ($d$,$k$)-tree topology, the adjacency matrix $A$ can be replaced by $d$ and $k$.
    The cutoff time $t_\mathrm{cut}$ is given by $p_\mathrm{gen}$, $F_\mathrm{new}$, $F_\mathrm{app}$, and $M$ via (\ref{eq.cutoffcondition}).
    }\label{tab.variables}
    \vspace{-2mm} 
\begin{tabular}{lp{0.8\textwidth}}
\multicolumn{2}{c}{\textbf{Physical topology}}\\
\hline
	$A$ & Physical adjacency matrix \\[5pt]
\multicolumn{2}{c}{\textbf{Hardware}}\\[1pt]
\hline
	$p_\mathrm{gen}$ & Probability of successful heralded entanglement generation \\
	$F_\mathrm{new}$ & Fidelity of newly generated entangled links \\
	$p_\mathrm{swap}$ & Probability of successful entanglement swap \\
	$r$ & Number of qubits per node per physical neighbor \\[5pt]
\multicolumn{2}{c}{\textbf{Software (application related)}}\\[1pt]
\hline
	$F_\mathrm{app}$  & Minimum fidelity to run background applications \\
	$M$ & Maximum number of short-distance links involved in a sequence of swaps \\
	$p_\mathrm{cons}$ & Probability that two nodes sharing some links consume one of them in each time slot \\[5pt]
\multicolumn{2}{c}{\textbf{CD protocol}}\\[1pt]
\hline
	$q$ & Probability of performing swaps according to the SRS protocol          
\end{tabular}

\end{table}

\section{Performance evaluation\\of CD protocols}\label{sec.virtual}
As previously discussed, a CD protocol must ensure that as many pairs of nodes as possible share entangled links, such that they can run quantum applications at any time. Ideally, the protocol should also provide many links between each pair of nodes, as this would allow them to run more demanding applications (e.g., applications that consume entanglement at a high rate) or to have spare links to run sporadic one-time applications.
These notions of a good CD protocol motivate the definition of the following performance metrics.

\begin{definition}
	In a quantum network, the \textit{virtual neighborhood} of node $i$, $V_i(t)$, is the set of nodes that share an entangled link with node $i$ at time $t$. Two nodes are \textit{virtual neighbors} if they share at least one entangled link. The \emph{virtual neighborhood size} is denoted as $v_i(t) \defeq |V_i(t)|$.
\end{definition}

\begin{definition}
	In a quantum network, the \textit{virtual node degree} of node $i$, $k_i(t)$, is the number of entangled links connected to node $i$ at time $t$.
\end{definition}

The virtual neighborhood size and the virtual node degree combined are useful metrics to evaluate the performance of a CD protocol.
The size of the virtual neighborhood of node $i$ corresponds to the number of nodes that can run background applications together with node $i$.
Since our model includes consumption of entanglement in such applications, the virtual degree provides information about how many resources are left to run sporadic applications.

The definitions above are similar to the notions of node neighborhood and node degree in classical graph theory. However, the configuration of entangled links changes over time, and therefore performance metrics from graph theory are ill-suited for this problem, as they generally do not include this type of time-dependence.
In contrast to those metrics, $v_i(t)$ and $k_i(t)$ are not random variables but stochastic processes, i.e., the value at each time slot is a random variable. 

When consuming entanglement at a constant rate, the steady state of the system is of particular interest since it will provide information about the performance of the protocols in the long term.
In Appendix \ref{app.analytical-steady-state}, we show that, when running the SRS protocol (Algorithm \ref{alg.srs}), the network undergoes a transient state and then reaches a unique steady-state regime (the proof also applies to similar CD protocols that use heralded entanglement generation, entanglement swaps, and cutoffs).
In what follows, we will focus on evaluating the performance of the protocol during the steady state via the steady-state expected value of the virtual neighborhood size, $v_i \equiv \lim_{t\rightarrow\infty} \mathbb{E}\big[ v_i (t) \big]$, and the virtual node degree, $k_i \equiv \lim_{t\rightarrow\infty} \mathbb{E}\big[ k_i (t) \big]$.

Next, in Subsection \ref{sec.results_noswaps}, we analyze the behavior of $v_i$ and $k_i$ in the absence of swaps.
In \ref{sec.results_optimization}, we analyze the relationship between these metrics and the protocol parameter $q$ in a tree-like network (although our methods are general and apply to any arbitrary topology) and we find the optimal $q$ that maximizes the virtual neighborhood size of the nodes in the lowest level of the tree.
In \ref{sec.results_hetero}, we provide a mathematical framework, based on Pareto optimization, to provide a good quality of service in heterogeneous networks.

\subsection{No swaps}\label{sec.results_noswaps}
To gain some intuition about the dynamics of the network and to set a benchmark, we consider the SRS protocol with $q=0$, i.e., no swaps. In the absence of swaps, only physical neighbors can share entanglement, and the virtual neighborhood size and the virtual node degree of node $i$ in the steady state are given by
\begin{equation}\label{eq.vi_noswaps_main}
	v_i \equiv \lim_{t\rightarrow\infty} \mathbb{E}\big[v_i(t) \big ] = d_i \frac{1-\frac{1-p_\mathrm{cons}}{1-p_\mathrm{gen}}\lambda^r}{1-\frac{p_\mathrm{cons}}{p_\mathrm{gen}}\lambda^r},
\end{equation}
\begin{equation}\label{eq.ki_noswaps_main}
	k_i \equiv \lim_{t\rightarrow\infty} \mathbb{E}\big[k_i(t) \big ] = d_i \, p_\mathrm{gen} \frac{r + \frac{p_\mathrm{cons}(1-p_\mathrm{cons})}{p_\mathrm{gen}-p_\mathrm{cons}}(\lambda^r-1)}{p_\mathrm{gen} - p_\mathrm{cons}\lambda^r},
\end{equation}
where $\lambda \equiv \frac{p_\mathrm{cons}(1-p_\mathrm{gen})}{p_\mathrm{gen}(1-p_\mathrm{cons})}$; $p_\mathrm{gen}$ is the probability of successful entanglement generation at the physical link level; $p_\mathrm{cons}$ is the link consumption rate; $d_i$ is the physical node degree of node $i$; and $r$ is the number of qubits available at node $i$ per physical neighbor.
(\ref{eq.vi_noswaps_main}) and (\ref{eq.ki_noswaps_main}) are derived in Appendix \ref{app.analytical-noswaps} using general random walks.
Note that in the derivation we assume large enough cutoffs, such that links are consumed with a high enough probability before reaching the cutoff time.

In the absence of swaps, both $v_i$ and $k_i$ are proportional to the physical node degree $d_i$ but independent of the rest of the physical topology. This allows us to study these performance metrics without assuming any specific physical topology.
Figure \ref{fig.analytical-no-swaps} shows the analytical solution for $v_i$ and $k_i$ when each node has five qubits per physical neighbor ($r=5$).
The figure shows a transition from large to small virtual neighborhood size when increasing $p_\mathrm{cons}$ beyond $p_\mathrm{gen}$. When the consumption rate is smaller than the generation rate, the size of the virtual neighborhood saturates and converges to the number of physical neighbors.
When $p_\mathrm{cons}$ increases beyond $p_\mathrm{gen}$, the virtual neighborhood size goes to zero. A similar behavior is observed for the virtual degree, which takes larger values for $p_\mathrm{cons}<p_\mathrm{gen}$.
The same behavior is observed for different values of $r$, as shown in Appendix~\ref{app.analytical-noswaps}.

We conclude that, when the consumption rate is below the generation rate and the cutoffs are large enough, each node can produce sufficient entangled links with its neighboring nodes for background applications and an extra supply of links for sporadic applications.

In Appendix \ref{app.analytical-noswaps}, we use simulations to show that $\mathbb{E}\big[ v_i (t)\big]$ and $\mathbb{E}\big[ k_i (t) \big]$ indeed converge to the steady-state values predicted by our analytical calculations as $t$ goes to infinity.

\begin{figure}[t]
\captionsetup[subfigure]{justification=centering}
\centering
  \centering
  \includegraphics[width=0.88\linewidth]{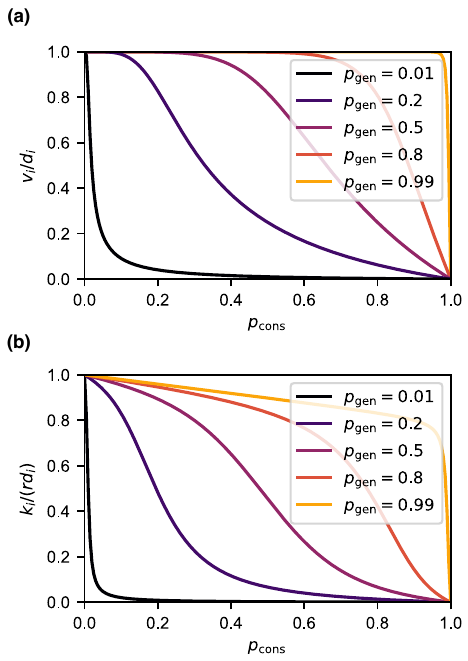}
\caption{\textbf{Larger consumption rates decrease the virtual neighborhood size and the virtual node degree.}
	Expected virtual neighborhood size \textbf{(a)} and virtual node degree \textbf{(b)} in the steady state in a quantum network with no swaps, with cutoff $t_\mathrm{cut}=10/p_\mathrm{cons}$ time steps, and with five qubits per node per physical channel ($r=5$). Both quantities are normalized by the physical degree of node $i$, $d_i$. The curves were calculated using (\ref{eq.vi_noswaps_main}) and (\ref{eq.ki_noswaps_main}).
	}
\label{fig.analytical-no-swaps}
\end{figure}

\subsection{Homogeneous set of users}\label{sec.results_optimization}
Let us now consider a more general setting: the SRS protocol with $q>0$. Nodes are now allowed to perform swaps with some probability $q$. In this setup, a natural question arises: what value of $q$ should we choose to achieve the best performance?

First, let us recall how we measure the performance. We use the expected virtual neighborhood size in the steady state, $v_i$, to determine the number of nodes that can run applications with node $i$ in the background.
We want to maximize $v_i$.
The expected node degree $k_i$ determines the number of additional entangled links that can be used for sporadic applications. Whenever possible, we will try to have a large $k_i$ too, although maximizing $k_i$ is not the purpose of a CD protocol (in fact, $k_i$ is maximized when no swaps are performed, since they always reduce the total number of entangled links, even when they are successful).
In what follows, we show how to optimize the SRS protocol in a quantum network with a $(2,3)$-tree topology (although our methods apply to any quantum network and any CD protocol).
This tree network is particularly interesting because it corresponds to a dumbbell network, which could be used to model users (level-2 nodes) in two metropolitan areas (level-1 nodes) connected by a central link via the level-zero node.
If we assume distances of the order of 10 km, the communication time over optical fibers is of the order of 1 ms.
Hence, the time step must be at least of the order of 1 ms.
For demonstration purposes, we assume a coherence time of $T = 2000$ time steps, which is of the order of 1 s.
As a reference, state-of-the-art coherence times lie between milliseconds (e.g., $T\approx 11.6$ ms in the NV centers experiment from ref. \cite{Pompili2021}) and seconds (e.g., $T\approx 50$ s in the trapped-ion experiment from ref. \cite{Harty2014}).
Additionally, also for demonstration purposes, we assume probabilistic entanglement generation, deterministic swaps, maximum swap distance $M=4$ (such that every node can share links with every other node), and background applications that can be executed with low fidelity links ($F_\mathrm{app}=0.6$).
We analyzed the system by simulating the evolution of the network over time and using Monte Carlo sampling.
For further details about how we find the steady state and how we compute expectation values from simulation data
see Appendix \ref{app.numerical-steady-state}.

Figure \ref{fig.performance_vs_q} shows $v_i$ and $k_i$ for three different nodes. Due to the symmetry of the topology, every node in the same level of the tree has the same statistical behavior. Therefore, we can describe the behavior of the whole tree network by looking at one node per level.
When no swaps are performed ($q=0$), the virtual neighborhood size $v_i$ (Figure~\ref{fig.performance_vs_q}a) is upper bounded by the number of physical neighbors $d_i$ ($d_i=2, 3,$ and $1$, for nodes in level 0, 1, and 2, respectively).
Increasing $q$ leads to an increase in $v_i$, which reaches a maximum value before decreasing again. If too many swaps are performed ($q$ close to 1), then $v_i$ decreases, since each swapping operation consumes two links and produces only one.
The maximum virtual neighborhood size, $\max_q v_i$, is achieved at a different value of $q$ for each node.
The virtual node degree $k_i$ (Figure~\ref{fig.performance_vs_q}b) behaves qualitatively in a similar way for every node: it is maximized at $q=0$ and, as we perform more swaps (increasing $q$), more links are swapped and fewer links remain in the system.
A similar behavior was observed for larger trees and for probabilistic swaps (see Appendix~\ref{app.extra}).

\begin{figure}[t]
\captionsetup[subfigure]{justification=centering}
\centering
  \centering
  \includegraphics[width=0.9\linewidth]{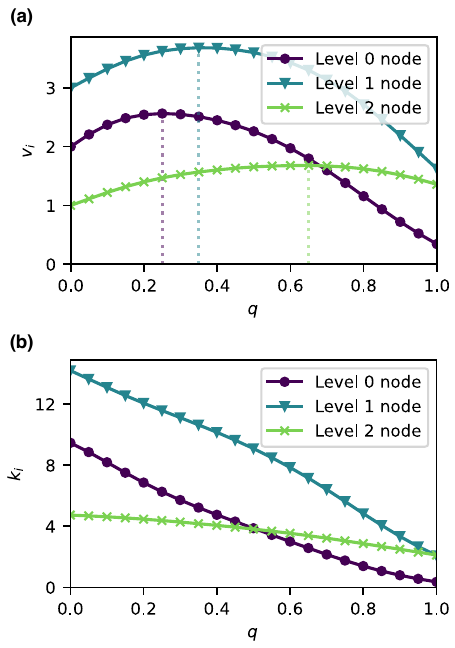}
\caption{\textbf{The virtual neighborhood size of every node cannot be maximized simultaneously.}
	Expected virtual neighborhood size \textbf{(a)} and virtual node degree \textbf{(b)} in the steady state in a (2,3)-tree network running the SRS protocol vs the protocol parameter~$q$.
	The value of $q$ that maximizes the virtual neighborhood size, indicated by the dotted lines, is node-dependent.
	The virtual node degree decreases monotonically with increasing $q$, since more links are consumed in swaps when $q$ is large.
	Other parameter values used in this experiment: $p_\mathrm{gen}=0.9$, $F_\mathrm{new} = 0.888$, $p_\mathrm{swap} = 1$, $r=5$, $T = 2000$ time steps, $M=4$, $p_\mathrm{cons}=p_\mathrm{gen}/4=0.225$, $F_\mathrm{app}=0.6$, $t_\mathrm{cut} = 56$ time steps (given by (\ref{eq.cutoffcondition})).
	Results obtained using a network simulation and Monte Carlo sampling with $10^6$ samples.
	Error bars are not shown since they are smaller than the line width -- the standard errors are below 0.003 and 0.006 for the $v_i$ and $k_i$, respectively.
	The standard error is defined as $2\hat\sigma/\sqrt{N_\mathrm{samples}}$, where $\hat\sigma$ is the sample standard deviation and $N_\mathrm{samples}$ is the number of samples.
	}
\label{fig.performance_vs_q}
\end{figure}

In some cases, we may be only interested in providing a good service to a subset of nodes $U$, the \emph{user nodes}.
The users run applications but also perform swaps to support the entanglement distribution among other pairs of users.
The only purpose of the rest of the nodes (\emph{repeater nodes}) is to aid the users to meet their needs.
In the literature, users that consume entanglement, but do not perform swaps to help other nodes, are generally called \emph{end nodes}. Here we assume every node is a user or a repeater node.
When some nodes are users and some are repeaters, the performance metrics of repeater nodes become irrelevant and we want to maximize $v_i$, $\forall i\in U$.
When the set of users is homogeneous (i.e., all user nodes have the same properties), the statistical behavior of all users is the same and we can formulate a single-objective optimization problem where we want to maximize $v_i$ for a single $i\in U$.
For example, in a tree quantum network, users are generally the nodes at the lowest level \cite{Choi2023}. In the example from Figure \ref{fig.performance_vs_q}, the lowest-level nodes are the level-2 nodes (green line with crosses).
If the level-2 nodes are the only users, the performance of the protocol is optimized for $q \approx 0.65$, which maximizes their $v_i$.
The protocol optimization problem also becomes a single-objective optimization problem in other networks with a strong symmetry, such as regular networks \cite{Talsma2023}.

In Figure \ref{fig.Fapp_pcons} we consider a (2,3)-tree network where the users are the nodes at the lowest level, and we study the influence of the background application requirements ($F_\mathrm{app}$ and $p_\mathrm{cons}$) on the maximum expected virtual neighborhood size of the users.
Here, we assume that the entanglement generation rate is much larger than the consumption rate ($p_\mathrm{gen} \geq 3p_\mathrm{cons}$). Otherwise, links are consumed shortly after they are generated and the behavior of the system is not interesting, as discussed in \ref{sec.results_noswaps}.
From the Figure, we observe that the consumption rate has a stronger effect on the virtual neighborhood.
For example, for $F_\mathrm{app}=0.8$, decreasing $p_\mathrm{cons}$ from 0.3 to 0.1 increases the maximum expected virtual neighborhood size by 20.3\%.
However, when decreasing $F_\mathrm{app}$ from 0.8 to 0.5, the maximum increase in $v_i$ is 3.5\% (for $p_\mathrm{cons}=0$).
The consumption rate has a bigger effect on the virtual neighborhood because it directly impacts the configuration of virtual links, while $F_\mathrm{app}$ only affects links via the cutoff.
In this case, the smallest cutoff is 17 time steps for $F_\mathrm{app}=0.8$ and the largest is 411 time steps for $F_\mathrm{app}=0.5$.
When the generation rate is large, virtual neighbors are likely to share multiple entangled links. In that case, cutoffs barely impact the virtual neighborhood size since links can be regenerated quickly and they are only removed after some time $t_\mathrm{cut}$. However, link consumption can still have a strong impact on the virtual neighborhood size since any link can be consumed at any time step.
If the cutoffs are very close to unity (e.g., when applications require a fidelity $F_\mathrm{app}>0.8$), the cutoff value may strongly affect the virtual neighborhood size.
\\

\begin{figure}[t]
\captionsetup[subfigure]{justification=centering}
\centering
  \centering
  \includegraphics[width=0.9\linewidth]{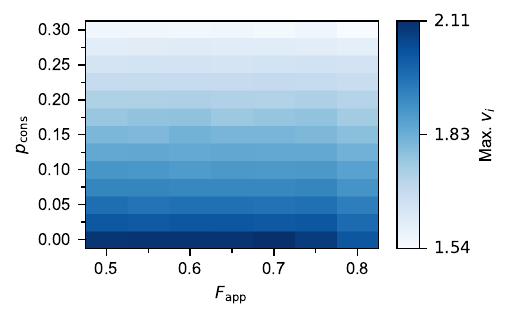}
\caption{\textbf{The consumption rate has a stronger impact on the performance than the application fidelity when the entanglement generation rate is high.}
	Maximum virtual neighborhood size (maximized over $q$) of a layer-2 node in a (2,3)-tree network vs the application fidelity, $F_\mathrm{app}$, and the consumption rate, $p_\mathrm{cons}$.
	Other parameter values used in this experiment: $p_\mathrm{gen}=0.9$, $F_\mathrm{new} = 0.95$, $p_\mathrm{swap} = 1$, $r=5$, $T = 2000$ time steps, $M=4$. The cutoff time $t_\mathrm{cut}$ is given by (\ref{eq.cutoffcondition}).
	Results obtained using a network simulation and Monte Carlo sampling with $10^4$ samples. The maximum error is 0.015 (the error is defined as $2\hat\sigma/\sqrt{N_\mathrm{samples}}$, where $\hat\sigma$ is the sample standard deviation and $N_\mathrm{samples}$ is the number of samples). Note that $\max_q(v_i)$ should be monotonic in $F_\mathrm{app}$ and $p_\mathrm{cons}$ but in this plot we observer small deviations due to the sample size.
	}
\label{fig.Fapp_pcons}
\end{figure}

\subsection{Heterogeneous set of users and\\multi-objective optimization}\label{sec.results_hetero}
In a more general topology, the user nodes may have different properties and different physical degrees.
In that case, the size of the virtual neighborhood of each user may be maximized for a different value of $q$.
Hence, optimizing the protocol for node $i$ generally means that the protocol will be suboptimal for some other node $j\neq i$.
This leads to a multi-objective optimization problem where we must find a tradeoff between the variables that we want to maximize.
In such a problem, optimality can be defined in different ways \cite{Marler2004}. A practical definition is the Pareto frontier:
\begin{definition}
	Let $U$ be the set of user nodes. Let $\vec\theta \in \Theta$ be a combination of parameter values describing the topology, the hardware, and the software of the quantum network, where $\Theta$ is the parameter space. Let $v_i(\vec\theta)$, with $i\in U$, be the set of variables that we want to maximize. The \emph{Pareto frontier} is defined as
	\begin{equation}
		P = \Big\{ \vec\theta \;\big|\; \forall \vec\theta'\in\Theta \;\, \exists i \;\mathrm{s.t.}\; v_i(\vec\theta) \geq v_i(\vec\theta') \Big\}.
	\end{equation}
\end{definition}

\begin{lemma}\label{lemma.pareto_nonempty}
	If the parameter space is non-empty, i.e., $\Theta\neq\emptyset$, then the Pareto frontier is non-empty, i.e., $P\neq\emptyset$.
\end{lemma}
\begin{proof}
	If $\Theta\neq\emptyset$, there exists some $\vec\theta_j = \mathrm{argmax}_{\vec\theta\in\Theta} \big(v_j(\vec\theta)\big)$, for any $j\in U$.
	Then, $v_j(\vec\theta_j) \geq v_j(\vec\theta')$, $\forall \vec\theta'	\in\Theta$, which means that $\vec\theta_j\in P$. Since $\vec\theta_j$ always exists, we conclude that $P\neq\emptyset$.
\end{proof}

Note that the parameter space $\Theta$ can be a constrained space, i.e., it does not necessarily include all combinations of parameters values. For example, combinations of parameters that are experimentally unfeasible may be excluded from~$\Theta$.
The Pareto frontier achieves a tradeoff in maximizing every $v_i$, $i\in U$. For all the points $\vec\theta$ in the Pareto frontier, we cannot obtain an increase in $v_i(\vec\theta)$ without decreasing or keeping constant some other $v_j(\vec\theta)$.
Moreover, Lemma~\ref{lemma.pareto_nonempty} ensures that there is at least one point $\vec\theta$ in the Pareto frontier.

Note that the Pareto frontier may allow situations in which the distribution of entangled links is not equitable (e.g., one user may maximize its virtual neighborhood size at the expense of another user minimizing it).
To avoid such situations, we can explicitly take into account quality-of-service requirements from every user node.
An example of simple requirement from node $i$ is to have some minimum number of virtual neighbors $c_i$. Then, the set of points that meet the quality-of-service requirements can be written as
\begin{equation}\label{eq.defQ}
	Q = \Big\{ \vec\theta \;\big|\; v_i(\vec\theta) \geq c_i \Big\}.
\end{equation}
An example of more specific requirement is to keep the number of entangled links between two specific nodes always above a certain threshold.

\begin{definition}
	The \emph{optimal region} $P^*$ is the set of parameters that are in the Pareto frontier and meet the quality-of-service requirements, i.e.,
	\begin{equation}
		P^* = P \cap Q,
	\end{equation}
	where $P$ is the Pareto frontier and $Q$ is the set of points that meet the quality-of-service requirements.
\end{definition}

As an example, we consider a $(2,3)$-tree network where the nodes in levels 1 and 2 are users.
Due to the symmetry of the topology, we only need to explicitly optimize $v_i$ for one node in each level.
In this case, it is possible to provide a graphical representation of the Pareto frontier and the optimal region.
Figure \ref{fig.pareto} shows the expected virtual neighborhood size in the steady state for a level-1 user and a level-2 user in a quantum network with a (2,3)-tree topology running the SRS protocol with probabilistic entanglement generation, deterministic swaps, and entanglement consumption at a fixed rate. Each data point corresponds to a different value of the protocol parameter $q$.
The data points highlighted with blue crosses form the Pareto frontier $P$.
In this example, we want the users in the first and second level to have an expected virtual neighborhood size larger than 3 and 1.6, respectively. Then,
\begin{equation}\label{eq.defQexample}
	Q = \Big\{ \vec\theta \;\big|\; v_1(\vec\theta) \geq 3, v_2(\vec\theta) \geq 1.6 \Big\}.
\end{equation}
The regions shaded in red correspond to forbidden regions where the quality-of-service requirements are not met. That is, the points in the white region are in $Q$.
The data points in the optimal region $P^*$ are the blue crosses in the white region. This corresponds to $q \in [0.4, 0.65]$.
All these values of $q$ can be considered optimal, as they are part of the Pareto frontier and meet the minimum user requirements.

As a final remark, note that we have used this multi-objective optimization framework to optimize the performance of a single-parameter CD protocol. However, it can also be used to choose from several CD protocols. This method can be applied to heterogeneous quantum networks with arbitrary topologies.

\begin{figure}[t!]
\captionsetup[subfigure]{justification=centering}
\centering
  \includegraphics[width=0.9\linewidth]{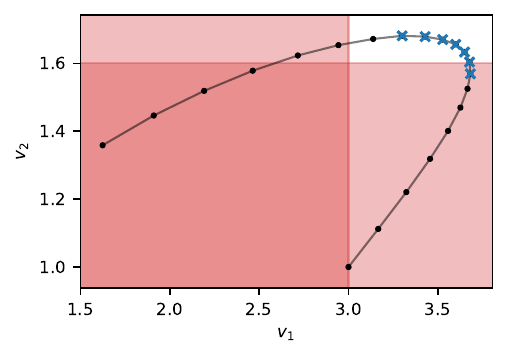}
\caption{\textbf{The optimal region determines the combinations of parameters that provide an optimal performance.}
	Virtual neighborhood size of a level-1 node, $v_1$, and a level-2 node, $v_2$, in a (2,3)-tree network running the SRS protocol for different values of the protocol parameter $q$ (for $q=0$, we have $v_1 = 3$ and $v_2=1$; we increase $q$ in intervals of 0.05 following the black line up to $q=1$).
	The data points with blue crosses form the Pareto frontier $P$. The regions shaded in red are forbidden by the quality-of-service requirements ($c_1=3$, $c_2=1.6$).
	The optimal region $P^*$ is formed by the blue crosses in the white region.
	Other parameter values used in this experiment: $p_\mathrm{gen}=0.9$, $F_\mathrm{new} = 0.888$, $p_\mathrm{swap} = 1$, $r=5$, $T = 2000$ time steps, $M=4$, $p_\mathrm{cons}=p_\mathrm{gen}/4=0.225$, $F_\mathrm{app}=0.6$, $t_\mathrm{cut} = 56$ time steps (given by (\ref{eq.cutoffcondition})).
	Results obtained using a network simulation and Monte Carlo sampling with $10^6$ samples.
	Error bars are not shown since they are smaller than the line width -- the standard errors are below 0.003 and 0.002 for $v_1$ and $v_2$, respectively.
	The standard error is defined as $2\hat\sigma/\sqrt{N_\mathrm{samples}}$, where $\hat\sigma$ is the sample standard deviation and $N_\mathrm{samples}$ is the number of samples.
	}
\label{fig.pareto}
\end{figure}

\section{Discussion}\label{sec.discussion}
In this paper we have introduced metrics to evaluate the performance of protocols for continuous distribution of entanglement.
The virtual neighborhood of a node is the set of nodes that share entanglement with the node, and the virtual degree of a node is the number of entangled states it shares with other nodes.
The goal of the protocol is to maximize the size of the virtual neighborhood of every user.
Here, as an example, we have considered a simple tree network and we have demonstrated how to formulate a single-objective and a multi-objective optimization problem that can be used to optimize the performance when the set of users is homogeneous and heterogeneous, respectively.

In our calculations, we assumed that background applications consume entanglement at a given rate.
We found that, when the entanglement generation rate is large, the consumption rate has a stronger impact on the size of the virtual neighborhood than the fidelity requirements imposed by the quantum applications.

Our formulation also allows the study of protocols that continuously distribute entanglement to maintain a supply of high quality pre-shared entanglement.
Specifically, the SRS protocol described in Algorithm \ref{alg.srs} delivers pre-shared entanglement when the consumption rate is set to zero.
This can be useful to determine the feasibility of quantum network protocols that assume pre-shared entanglement among the nodes of the network.
In this case, an application that uses the available entanglement would disrupt the distribution of entangled states and would bring the system to a new transient state. Hence, an additional useful metric would be the time required to converge to a steady state after such a disruption.
We leave this analysis for future work.

We also leave the generalization of the network model and the protocol as future work.
As an example, one can consider nodes that have a pool of qubits that can be used for any purpose, instead of having $r$ specific qubits that can generate entanglement with each physical neighbor.
One can also define node-dependent protocols, where each node follows a different set of instructions.

Lastly, note that we expect the coexistence of protocols for on-demand and continuous distribution of entanglement in large-scale quantum networks. Continuous distribution can be used to supply entanglement to applications running at a constant rate while on-demand distribution can support this process during peak demands from sporadic applications.

\section{Data availability}\label{sec.data_availability}
The data shown in this paper can be found in ref. \cite{Inesta2023a-data}.

\section{Code availability}\label{sec.code_availability}
Our code can be found in the following GitHub repository: \href{https://github.com/AlvaroGI/optimizing-cd-protocols}{https://github.com/AlvaroGI/optimizing-cd-protocols}.
\vspace{40pt}


\section{Acknowledgments}
We thank B. Davies, T. Coopmans, and G. Vardoyan for discussions and feedback.
We also thank J. van Dam, B. van der Vecht, and L. Talsma for feedback on this manuscript.
\'{A}GI acknowledges financial support from the Netherlands Organisation for Scientific Research (NWO/OCW), as part of the Frontiers of Nanoscience program.
SW acknowledges support from an ERC Starting Grant.

\section{Author contributions}
\'{A}GI defined the project, analyzed the results, and prepared this manuscript.
SW supervised the project and provided active feedback at every stage of the project.

\section{Competing interests}
The authors declare no competing interests.

\section{Additional information}
\textbf{Supplementary information} is available at the end of this document.
\\

\textbf{Correspondence} should be addressed to Álvaro G. Iñesta.

\clearpage
\onecolumngrid
\appendix

\section{Further details on the network model}\label{app.model}
\ssec{Entanglement swap.}
Two nodes that are not physical neighbors cannot generate entanglement directly between them. Instead, they rely on entanglement swap operations to produce a shared entangled state between them \cite{Zukowski1993, Duan2001, Sangouard2011}.
As an example, consider two end nodes A and B, which are not physically connected but share a physical link with an intermediate node C.
To generate an entangled link between A and B, they need to first generate entangled links between A and C, and also between C and B.
Then, node C can perform a Bell state measurement to transform links A-C and C-B into a single entangled link between A and B.
When both input links are Werner states with fidelities $F_1$ and $F_2$, the output state in a swap operation is also a Werner state with fidelity \cite{Munro2015}
\begin{equation}
	F_\text{swap}(F_1,F_2) = F_1\cdot F_2 + \frac{(1-F_1)\cdot(1-F_2)}{3}.
\end{equation}
Note that this operation generally decreases the fidelity: $F_\text{swap}(F_1,F_2) \leq F_1, F_2$.

Additionally, entanglement swaps can be either probabilistic \cite{Calsamiglia2001, Duan2001, Ewert2014} or deterministic \cite{Pompili2021}, depending on the hardware employed.
With probability $p_\text{s}$, the swap operation succeeds and both input states are consumed to produce a single entangled link. With probability $1-p_\text{s}$, the swap operation fails: both input states are consumed but no other entangled state is produced.
\\

\ssec{Purification.}
If the fidelity of an entangled state is not large enough for a specific application, nodes can run a purification protocol to increase its fidelity. In general, these protocols take as input multiple entangled states and output a single state with larger fidelity \cite{Dur2007, Hartmann2007, Victora2020}.

For simplicity, we do not consider any kind of purification in our analysis.
Nevertheless, it is possible to integrate purification of entangled states at the physical link level into our model by decreasing the value of $p_\mathrm{gen}$, to account for all the states that must be prepared in advance to perform the purification protocol. This would also impact the fidelity of newly generated links, $F_\mathrm{new}$, which would correspond now to the fidelity of the links after purification at the physical link level. The cutoff time would also need to be adjusted, since the time step would take a longer time (it would have to include more than one entanglement generation attempt).
If physically distant nodes require larger fidelity links, they can run a purification subroutine as part of the application once they have generated enough entangled links.
\\

\ssec{Cutoff times.}
Quantum states decohere, mainly due to environmental couplings \cite{Dur2005,Chirolli2008}. Decoherence decreases the fidelity of states over time.
We consider a depolarizing noise model, which is a worst-case scenario (other types of noise can be converted to depolarizing noise via twirling \cite{Horodecki1999,Dur2005,Cai2019}).
As shown in Appendix A from \cite{Inesta2023}, if we assume that each qubit of a Werner state is stored in a different memory and experiences depolarizing noise independently, the fidelity of the Werner states evolves as
\begin{equation}\label{eq.Fwerner}
	F(t+\Delta t) = \frac{1}{4} + \bigg(F(t)-\frac{1}{4} \bigg)e^{-\frac{\Delta t}{T}},
\end{equation}
where $F(t)$ is the fidelity of the state at time $t$, $\Delta t$ is an arbitrary interval of time, and $T$ is a parameter that characterizes the exponential decay in fidelity of the whole entangled state.

When the fidelity of the entangled links drops below some threshold, they are no longer useful.
Hence, a common practice is to discard states after a cutoff time $t_\mathrm{cut}$ to prevent wasting resources on states that should not be used anymore \cite{Rozpedek2018}.
We refer to the time passed since the creation of a quantum state as the age of the state.
Whenever the age of an entangled state equals the cutoff time, the state is removed, i.e., the qubits involved are reset.
As shown in \cite{Inesta2023}, to ensure that any two nodes that are at most $M$ physical links away will only share entangled states with fidelity larger than $F_\text{app}$, the cutoff time must satisfy
\begin{equation}\label{eq.cutoffcondition_appendix}
    t_\mathrm{cut} \leq -T \ln\Bigg(\frac{3}{4F_\text{new}-1} \Big( \frac{4F_\text{app}-1}{3} \Big)^{1/M} \Bigg),
\end{equation}
where $F_\text{new}$ is the fidelity of newly generated entangled links.
This condition assumes that the output state in a swap operation takes the age of the oldest input link.
\vspace{20pt}

\clearpage
\section{Existence of a unique steady state}\label{app.analytical-steady-state}
In this Appendix, we show that there is a unique steady-state value for the expected number of virtual neighbors and expected virtual degree of any node when a quantum network is running CD Protocol \ref{alg.srs}, under the assumption that entanglement generation is probabilistic ($p_\mathrm{gen}<1$).

We consider the stochastic processes $v_i(t)$ and $k_i(t)$, which correspond to the number of virtual neighbors of node $i$ and the virtual degree of node $i$, respectively.
The expected values over many realizations of the processes are denoted as $\mathbb{E}\big[ v_i(t) \big]$ and $\mathbb{E}\big[ k_i(t) \big]$.

The state of the network can be represented using the ages of all entangled links present in the network (the age is measured in number of time slots). This can be written as an array $s$ with $\frac{1}{2} r \sum_{i=0}^{n-1} d_i$ components, since there are $n$ nodes and each node $i$ can store up to $rd_i$ entangled links, where $d_i$ is the physical degree of node $i$ and $r$ is a hardware-dependent parameter that limits the maximum number of qubits per node.
Since we impose cutoff times on the memories, each of the components of this vector can only take a finite set of values.
Let $\mathcal{S}$ be the set of all possible states, which is also finite.

Given a state $s(t)$ at time $t\in\mathcal{N}$ (recall that we consider discrete time steps in our protocols), the transition to a new state only depends on the number of available memories at each node for generation of new links and on the number of available links for performing swaps and for consumption in applications.
Hence, the transition does not depend on past information:
$$\text{Pr}\big[ s(t+1)=\sigma \;|\; s(0), s(1), \dots, s(t) \big] 
= \text{Pr}\big[ s(t+1)=\sigma \;|\; s(t) \big].$$

Consequently, the state of the network can be modeled as a Markov chain with the following three properties:
\begin{enumerate}
	\item The chain is irreducible, since every state is reachable from every other state.
	If $p_\mathrm{gen}<1$, there is a nonzero probability that no links are generated over many time slots until all existing links expire due to cutoffs and therefore the network returns to the starting state with no links -- from this initial state, every other state can be reached.
	\item The chain is aperiodic. A sufficient condition for an irreducible chain to be aperiodic is that $\text{Pr}\big[ s(t+1)=\sigma \;|\; s(t)=\sigma \big] > 0$ for some state $\sigma\in\mathcal{S}$ \cite{VanMieghem2014}. When entanglement generation is probabilistic ($p_\mathrm{gen}<1$), the state with no entangled links satisfies the previous condition (if all entanglement generation attempts fail, the network will remain in a state with no links), and therefore the chain is aperiodic.
	\item The chain is positive recurrent (i.e., the mean time to return to any state is finite), since it is irreducible and it has a finite state space $\mathcal{S}$ (see Theorem 9.3.5 from Ref. \cite{VanMieghem2014}).
\end{enumerate}

According to Theorem 9.3.6 from Ref. \cite{VanMieghem2014}, from the three properties above we can conclude that there exists a unique steady-state probability distribution, i.e., the following limit exists: $\lim_{t\rightarrow\infty} \text{Pr}\big[ s(t)=\sigma \big]$, $\forall \sigma \in \mathcal{S}$.

Let us now compute the expected number of virtual neighbors in the steady state:
\begin{equation}\label{eq.v.ss.long}
\begin{split}
	v_i \equiv \lim_{t\rightarrow\infty} \mathbb{E}\big[ v_i(t) \big]
	&= \lim_{t\rightarrow\infty} \sum_{v=0}^{n} v\cdot\mathrm{Pr} \big[ v_i(t)=v \big] \\
	&= \sum_{v=0}^{n} v\cdot \lim_{t\rightarrow\infty} \mathrm{Pr} \big[ v_i(t)=v \big] \\
	&= \sum_{v=0}^{n} v\cdot \lim_{t\rightarrow\infty} \sum_{\sigma\in\mathcal{S}} \mathrm{Pr} \big[ v_i(t)=v \;|\; s(t)=\sigma \big] \cdot \mathrm{Pr} \big[ s(t)=\sigma \big] \\
	&= \sum_{v=0}^{n} v\cdot \sum_{\sigma\in\mathcal{S}} \lim_{t\rightarrow\infty} \mathrm{Pr} \big[ v_i(t)=v \;|\; s(t)=\sigma \big] \cdot \lim_{t\rightarrow\infty} \mathrm{Pr} \big[ s(t)=\sigma \big] \\
	&= \sum_{\sigma\in\mathcal{S}} \lim_{t\rightarrow\infty} \mathrm{Pr} \big[ s(t)=\sigma \big] \cdot \sum_{v=0}^{n} v \cdot \lim_{t\rightarrow\infty} \mathrm{Pr} \big[ v_i(t)=v \;|\; s(t)=\sigma \big].
\end{split}
\end{equation}
Let us define a function $\kappa(s;i,j)$ that takes as input a state $s$ and two node indices $i$ and $j$. This function returns the number of entangled links shared by nodes $i$ and $j$ in state $s$.
The virtual neighborhood size of node $i$ at time $t$, $v_i(t)$, is given by the state of the network at time $t$, $s(t)$, and it can be written as $$v_i(t) = v_i(s(t)) = \sum_{j \in V \setminus \{i\}} \min\Big(1,\kappa\big(s(t);i,j\big)\Big).$$
Consequently,
\begin{equation}\label{eq.pr.v}
	\mathrm{Pr} \big[ v_i(t)=v \;|\; s(t)=\sigma \big] = 
	\begin{cases}
	1,\;\;\mathrm{if}\;\;v = \sum_{j \in V \setminus \{i\}} \min\big(1,\kappa(\sigma;i,j)\big)\\
	0,\;\;\mathrm{otherwise}
	\end{cases}.
\end{equation}
Using (\ref{eq.pr.v}), we can write (\ref{eq.v.ss.long}) as
\begin{equation}\label{eq.v.ss}
	v_i
	= \sum_{\sigma\in\mathcal{S}} \lim_{t\rightarrow\infty} \mathrm{Pr} \big[ s(t)=\sigma \big] \sum_{j \in V \setminus \{i\}} \min\big(1,\kappa(\sigma;i,j)\big),
\end{equation}

The expected virtual degree can be calculated similarly but using its corresponding definition, $k_i(s(t)) = \sum_{j \in V \setminus \{i\}} \kappa\big(s(t);i,j\big)$:
\begin{equation}\label{eq.k.ss}
	k_i \equiv \lim_{t\rightarrow\infty} \mathbb{E}\big[ k_i (t) \big]
	= \sum_{\sigma\in\mathcal{S}} \lim_{t\rightarrow\infty}  \mathrm{Pr} \big[ s(t)=\sigma \big] \sum_{j \in V \setminus \{i\}} \kappa\big(\sigma;i,j\big).
\end{equation}

Since we have shown that the probability distributions that appear in (\ref{eq.v.ss}) and (\ref{eq.k.ss}) exist and are unique, then the quantities $v_i$ and $k_i$ also exist and are unique. That is, there is a unique steady-state value for the expected number of virtual neighbors and the expected virtual degree of any node $i$.

From our simulations, we also expect a unique steady state for $p_\mathrm{gen}=1$.
The main difficulty in proving its existence is that the Markov chain is not always irreducible (the state with no links may not be reachable from some other states since links are generated at maximum rate). However, if one can show that there is a unique equivalence class (i.e., a unique set of states that are reachable from each other) that is reached after a finite number of transitions, the derivation above may be applicable to this equivalence class, which would constitute an irreducible Markov chain.
\\

Lastly, note that in practice one may find an initial transient state with periodic behavior.
This happens in quasi-deterministic systems, i.e., systems in which all probabilistic events (e.g., successful entanglement generation) happen with probability very close to 1.
In quasi-deterministic systems, all realizations of the stochastic processes are identical at the beginning with a very large probability. For some combinations of parameters, these processes may display a periodic behavior with a period on the order of the cutoff time.
Over time, each realization starts to behave differently due to some random events yielding different outcomes.
Consequently, the periodic oscillations will dephase, and they will cancel out after averaging over all realizations.
In the example from Figure \ref{fig.periodic_steady_state}, we find that both $\mathbb{E}\big[ v_i(t) \big]$ and $\mathbb{E}\big[ k_i(t) \big]$ are periodic with period approximately $t_\mathrm{cut}$. The amplitude of the oscillations vanishes after a few periods.
\\

\begin{figure}[h!]
\captionsetup[subfigure]{justification=centering}
     \centering
     \begin{subfigure}[b]{0.4\textwidth}
         \centering
         \includegraphics[width=\textwidth]{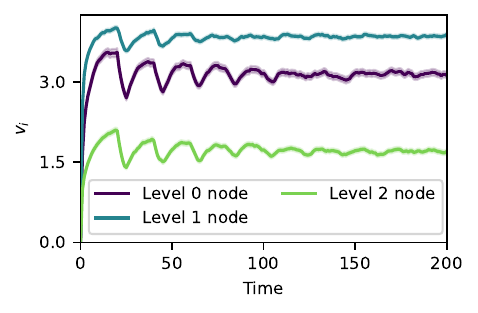}
	\vspace{-1.5\baselineskip}
         \caption{Expected virtual neighborhood size.}
         \vspace{5pt}
     \end{subfigure}
     \begin{subfigure}[b]{0.4\textwidth}
         \centering
         \includegraphics[width=\textwidth]{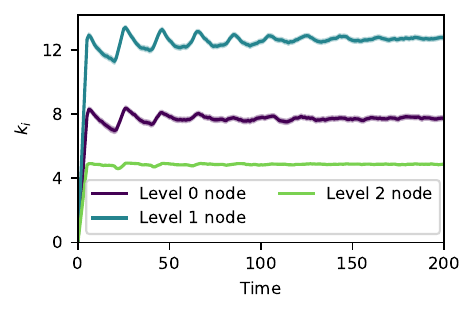}
	\vspace{-1.5\baselineskip}
         \caption{Expected virtual node degree.}
         \vspace{5pt}
     \end{subfigure}

	\caption{\textbf{A transient state with periodic oscillations may exist in quasi-deterministic systems.}
	Evolution of $v_i$ and $k_i$ in a quantum network with a $(2,3)$-tree topology running the SRS protocol described in the main text.
	Each line (purple, blue, and green) corresponds to a node in a different level of the tree (level 0, 1, and 2).
	The error for each solid line is shown as a shaded region, although it is hard to notice since its maximum value is 0.040 in \textbf{(a)} and 0.084 in \textbf{(b)} (the error is defined as $2\hat\sigma/N_\mathrm{samples}$, where $\hat\sigma$ is the sample standard deviation and $N_\mathrm{samples}$ is the number of samples).
	Other parameters used in this experiment: $p_\mathrm{gen}=0.99$, $F_\mathrm{new} = 0.88$, $p_\mathrm{swap}=1$, $r=5$, $T = 2000$ time steps, $M=4$, $p_\mathrm{cons}=0.01$, $q=0.2$, $F_\mathrm{app}=0.6$, $t_\mathrm{cut} = 20$ time steps.
	Numerical results obtained using a network simulation and Monte Carlo sampling with $10^3$ samples.
	}
	\label{fig.periodic_steady_state}
\end{figure}
\vspace{20pt}

\clearpage
\section{Analytical performance metrics in the absence of swaps}\label{app.analytical-noswaps}
In this Appendix, we consider a CD protocol with the same structure as the SRS protocol (see Algorithm \ref{alg.srs} from the main text) in the absence of swaps and with a large enough cutoff time ($t_\mathrm{cut}>r$ and $t_\text{cut}\gg\frac{1}{p_\mathrm{cons}}$, where the cutoff is measured in number of time steps).
As discussed in the main text, when no swaps are performed, we can derive closed-form expressions to gain some intuition about the dynamics of the network and to set a benchmark.
Here, we show that the virtual neighborhood size and the virtual node degree of node $i$ in the steady state are given by
\begin{equation}\label{eq.vi_noswaps_app}
	v_i \equiv \lim_{t\rightarrow\infty} \mathbb{E}\big[v_i(t) \big ] = d_i \frac{1-\frac{1-p_\mathrm{cons}}{1-p_\mathrm{gen}}\lambda^r}{1-\frac{p_\mathrm{cons}}{p_\mathrm{gen}}\lambda^r}
\end{equation}
and
\begin{equation}\label{eq.ki_noswaps_app}
	k_i \equiv \lim_{t\rightarrow\infty} \mathbb{E}\big[k_i(t) \big ] = d_i \, p_\mathrm{gen} \frac{r + \frac{p_\mathrm{cons}(1-p_\mathrm{cons})}{p_\mathrm{gen}-p_\mathrm{cons}}(\lambda^r-1)}{p_\mathrm{gen} - p_\mathrm{cons}\lambda^r},
\end{equation}
where $\lambda \equiv \frac{p_\mathrm{cons}(1-p_\mathrm{gen})}{p_\mathrm{gen}(1-p_\mathrm{cons})}$; $p_\mathrm{gen}$ is the probability of successful entanglement generation at the physical link level; $p_\mathrm{cons}$ is the link consumption probability; $d_i$ is the physical node degree of node $i$; and $r$ is the number of qubits per physical link available at each node.

We define $w_{ij}$ as the number of entangled links shared between nodes $i$ and $j$ (similar to the definition of $\kappa$ in Appendix~\ref{app.analytical-steady-state}).
In the absence of swaps, nodes $i$ and $j$ can only share entangled links if they are physical neighbors, since the only mechanism available is heralded entanglement generation. If $i$ and $j$ are not physical neighbors, then $w_{ij}=0$.
The entangled links shared between nodes $i$ and $j$ can be consumed in some application or discarded when applying cutoffs.
However, we also assume that entangled links are always consumed before they reach the cutoff time, i.e., $t_\text{cut}\gg\frac{1}{p_\mathrm{cons}}$ (cutoff measured in number of time steps).
This assumption allows us to model $w_{ij}$ using the general random walk shown in Figure~\ref{fig.random-walk}:
\begin{itemize}
	\item The maximum value for $w_{ij}$ is the number of qubits available per physical link, $r$. This state is reachable even when entanglement generation is done sequentially, since links can be stored for longer than $r$ time steps (we assume $t_\mathrm{cut} > r$).
	\item The probabilities of transition forward are $p_k = p_\mathrm{gen}(1-p_\mathrm{cons})$, $\forall k<r$, and $p_r=0$.
	\item The probabilities of transition backwards are $q_0=0$, $q_k = p_\mathrm{cons}(1-p_\mathrm{gen})$, $\forall 0<k<r$, and $q_r = p_\mathrm{cons}$.
	\item The no-transition probability is $z_k = 1-p_k-q_k$, $\forall k$.
\end{itemize}

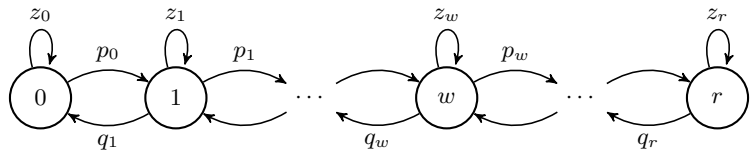
\begin{figure}[h]
	\begin{tikzpicture}[->, >=stealth', auto, semithick, node distance=1.8cm]
	\tikzstyle{every state}=[fill=white,draw=black,thick,text=black,scale=1]
	\node[state]    (0)               {$0$};
	\node[state]    (1)[right of=0]   {$1$};
	\node[]    (dots1)[right of=1]   {$\cdots$};
	\node[state]    (w)[right of=dots1]   {$w$};
	\node[]    (dotsw)[right of=w]   {$\cdots$};
	\node[state]    (g)[right of=dotsw]   {$r$};
	\path (0) edge [bend left] node [above] {$p_0$} (1); 
        \path (1) edge [bend left] node [above] {$p_1$} (dots1);
        \path (dots1) edge [bend left] node [above] {} (w);
        \path (w) edge [bend left] node [above] {$p_w$} (dotsw);
        \path (dotsw) edge [bend left] node [above] {} (g);
        \path (0) edge [loop above] node [above] {$z_0$} (0);
        \path (1) edge [loop above] node [above] {$z_1$} (1);
        \path (w) edge [loop above] node [above] {$z_w$} (w);
        \path (g) edge [loop above] node [above] {$z_r$} (g);
	\path (g) edge [bend left] node [below] {$q_r$} (dotsw);
	\path (dotsw) edge [bend left] node [below] {} (w);
	\path (w) edge [bend left] node [below] {$q_w$} (dots1);
	\path (dots1) edge [bend left] node [below] {} (1);
	\path (1) edge [bend left] node [below] {$q_1$} (0);
	\end{tikzpicture}
\caption{\textbf{General random walk modeling the number of entangled links $w_{ij}$ between nodes $i$ and $j$ in the absence of swaps.}}
\label{fig.random-walk}
\end{figure}

The steady-state probability distribution of this Markov chain is given by \cite{VanMieghem2014}
\begin{equation}
	\lim_{t\rightarrow\infty} \text{Pr}\big[w_{ij}(t)=w\;|\;A_{ij}=1\big] = 
	\begin{cases}
	\Big( 1 + \sum_{k=1}^{r}\prod_{m=0}^{k-1}\frac{p_m}{q_{m+1}} \Big) ^{-1}, \;\;w=0\\
	\Big( 1 + \sum_{k=1}^{r}\prod_{m=0}^{k-1}\frac{p_m}{q_{m+1}} \Big) ^{-1} \prod_{m=0}^{w-1}\frac{p_m}{q_{m+1}}, \;\;w>0
	\end{cases},
\end{equation}
where $A_{ij}$ is a binary variable that indicates if nodes $i$ and $j$ are physical neighbors ($A_{ij}=1$) or not ($A_{ij}=0$).
After some algebra, the previous equation can be rewritten in terms of the original variables of the problem:
\begin{equation}\label{eq.limprobwij}
	\lim_{t\rightarrow\infty} \text{Pr}\big[w_{ij}(t)=w\;|\;A_{ij}=1\big] = 
	\begin{cases}\vspace{5pt}
	\pi_0, \;\;w=0\\\vspace{5pt}
	\pi_0 \rho^w, \;\;0<w<r\\
	\pi_0 \rho^r (1-p_\mathrm{gen}), \;\;w=r
	\end{cases},
\end{equation}
where
\begin{equation}\label{eq.pi0andrho}
	\pi_0 \equiv \frac{p_\mathrm{gen}-p_\mathrm{cons}}{(1-p_\mathrm{gen}) ( p_\mathrm{gen}\rho^r -p_\mathrm{cons} )}
\;\;\;\;\text{and}\;\;\;\;
\rho \equiv \frac{p_\mathrm{gen}(1-p_\mathrm{cons})}{p_\mathrm{cons}(1-p_\mathrm{gen})}.
\end{equation}

The expected value of $w_{ij}$ is
\begin{equation}\label{eq.limexpwij}
\begin{split}
	\lim_{t\rightarrow\infty} \mathbb{E}\big[w_{ij}(t)\;|\;A_{ij}=1\big]
	&= \lim_{t\rightarrow\infty} \sum_{w=0}^r w \cdot \text{Pr}\big[w_{ij}(t)=w\;|\;A_{ij}=1\big]\\
	&= \sum_{w=0}^r w \cdot \lim_{t\rightarrow\infty}\text{Pr}\big[w_{ij}(t)=w\;|\;A_{ij}=1\big]\\
	&\stackrel{a}{=} \pi_0 \sum_{w=0}^{r-1} w \rho^w + r \pi_0 \rho^r (1-p_\mathrm{gen})\\
	&= \pi_0 \frac{\rho - r\rho^r + (r-1)\rho^{r+1}}{(1-\rho)^2} + r \pi_0 \rho^r (1-p_\mathrm{gen})\\
	&= \frac{p_\mathrm{gen}}{(p_\mathrm{gen}-p_\mathrm{cons})(p_\mathrm{gen}\rho^r - p_\mathrm{cons})} \Big( r(p_\mathrm{gen}-p_\mathrm{cons})\rho^r + p_\mathrm{cons}(1-p_\mathrm{cons})(1-\rho^r)\Big),
\end{split}
\end{equation}
where we have used (\ref{eq.limprobwij}) in step $a$.

The virtual neighborhood size of node $i$ is defined in terms of the variables $w_{ij}$ as $v_i(t) = \sum_{j=1}^{n} \min(w_{ij}(t),1)$, and the expectation value can be calculated as follows:
\begin{equation}\label{eq.expected_vi_noswaps}
\begin{split}
	v_i \equiv \lim_{t\rightarrow\infty} \mathbb{E}\big[v_i(t) \big]
	&= \lim_{t\rightarrow\infty} \mathbb{E}\bigg[\sum_{j=1}^{n} \min(w_{ij}(t),1) \bigg]
	= \sum_{j=1}^{n} \lim_{t\rightarrow\infty} \mathbb{E}\big[\min(w_{ij}(t),1) \big]\\
	&\stackrel{a}{=} \sum_{j=1}^{n} \lim_{t\rightarrow\infty} \sum_{x=0}^{r} x \cdot \text{Pr}\big[\min(w_{ij}(t),1)=x\big]
	= \sum_{j=1}^{n} \sum_{x=0}^{r} x \cdot \lim_{t\rightarrow\infty} \text{Pr}\big[\min(w_{ij}(t),1)=x\big]\\
	&\stackrel{b}{=} \sum_{j=1}^{n} \lim_{t\rightarrow\infty} \text{Pr}\big[\min(w_{ij}(t),1)=1\big]
	= \sum_{j=1}^{n} \lim_{t\rightarrow\infty} \text{Pr}\big[w_{ij}(t)>0\big]\\
	&\stackrel{c}{=} \sum_{j=1}^{n} \lim_{t\rightarrow\infty} \text{Pr}(A_{ij}=1) \cdot \text{Pr}\big[w_{ij}(t)>0 \;|\;A_{ij}=1 \big]\\
	&\stackrel{d}{=} \sum_{j=1}^{n} A_{ij} \lim_{t\rightarrow\infty} \Big(1-\text{Pr}\big[w_{ij}(t)=0\;|\;A_{ij}=1\big]\Big)\\
	&\stackrel{e}{=} \sum_{j=1}^{n} A_{ij}(1-\pi_0)\\
	&\stackrel{f}{=} d_i(1-\pi_0)\\
	&\stackrel{g}{=} d_i \frac{p_\mathrm{gen}^{r+1} (1-p_\mathrm{cons})^r - p_\mathrm{gen}(1-p_\mathrm{gen})^{r-1} p_\mathrm{cons}^r (1-p_\mathrm{cons})}{p_\mathrm{gen}^{r+1} (1-p_\mathrm{cons})^r - (1-p_\mathrm{gen})^{r} p_\mathrm{cons}^{r+1}}\\
	&= d_i \frac{1-\frac{1-p_\mathrm{cons}}{1-p_\mathrm{gen}}\lambda^r}{1-\frac{p_\mathrm{cons}}{p_\mathrm{gen}}\lambda^r}, 
\end{split}
\end{equation}
where $\lambda \equiv \frac{p_\mathrm{cons}(1-p_\mathrm{gen})}{p_\mathrm{gen}(1-p_\mathrm{cons})}$, $d_i$ is the physical degree of node $i$, and $n$ is the total number of nodes, and with the following steps:
\begin{enumerate}[label=\alph*.]
	\item We use the definition of expected value and the fact that $w_{ij}(t) \leq r$.
	\item We use the fact that $\min(w_{ij}(t),1) \in \{0, 1\}$.
	\item We use the law of total probability, i.e., $\text{Pr}(X) = \sum_n \text{Pr}(Y_n) \cdot \text{Pr}(X|Y_n)$. Moreover, if two nodes $i$ and $j$ are not physical neighbors ($A_{ij}=0$), they cannot share any entangled links due to the absence of swaps, i.e., $\text{Pr}\big[w_{ij}(t)>0 \;|\;A_{ij}=0 \big] = 0$.
	\item Given the topology, $A_{ij}$ is a binary variable with a fixed value. Therefore, $\text{Pr}(A_{ij}=1) = A_{ij}$.
	\item We use (\ref{eq.limprobwij}).
	\item The physical node degree of node $i$ can be computed as $d_i = \sum_{j=1}^{n} A_{ij}$.
	\item We use (\ref{eq.pi0andrho}).
\end{enumerate}

The virtual degree of node $i$ is defined in terms of the variables $w_{ij}$ as $k_i(t)=\sum_{j=1}^{n} w_{ij}(t)$, and the expectation value can be calculated in a similar way to $v_i$:
\begin{equation}\label{eq.expected_ki_noswaps}
\begin{split}
	k_i \equiv \lim_{t\rightarrow\infty} \mathbb{E}\big[k_i(t) \big]
	&= \lim_{t\rightarrow\infty} \mathbb{E}\bigg[\sum_{j=1}^{n} w_{ij}(t) \bigg]
	= \lim_{t\rightarrow\infty} \sum_{j=1}^{n} \mathbb{E}\big[w_{ij}(t) \big]\\
	&\stackrel{a}{=} \lim_{t\rightarrow\infty} \sum_{j=1}^{n} A_{ij} \cdot \mathbb{E}\big[w_{ij}(t) \;|\; A_{ij}=1\big]\\
	&\stackrel{b}{=} \lim_{t\rightarrow\infty} \mathbb{E}\big[w_{ij}(t) \;|\; A_{ij}=1\big] \cdot \sum_{j=1}^{n} A_{ij}\\
	&\stackrel{c}{=} d_i \cdot \lim_{t\rightarrow\infty} \mathbb{E}\big[w_{ij}(t) \;|\; A_{ij}=1\big]\\
	&\stackrel{d}{=} d_i  p_\mathrm{gen}  \frac{ r (p_\mathrm{gen}-p_\mathrm{cons}) p_\mathrm{gen}^r (1-p_\mathrm{cons})^r + p_\mathrm{cons}(1-p_\mathrm{cons})\big(p_\mathrm{cons}^r(1-p_\mathrm{gen})^r - p_\mathrm{gen}^r(1-p_\mathrm{cons})^r\big) }{ (p_\mathrm{gen}-p_\mathrm{cons})\big(p_\mathrm{gen}^{r+1}(1-p_\mathrm{cons})^r - p_\mathrm{cons}^{r+1} (1-p_\mathrm{gen})^r \big)}\\
	&= d_i  p_\mathrm{gen} \frac{r + \frac{p_\mathrm{cons}(1-p_\mathrm{cons})}{p_\mathrm{gen}-p_\mathrm{cons}}(\lambda^r-1)}{p_\mathrm{gen} - p_\mathrm{cons}\lambda^r},
\end{split}
\end{equation}
where $\lambda \equiv \frac{p_\mathrm{cons}(1-p_\mathrm{gen})}{p_\mathrm{gen}(1-p_\mathrm{cons})}$, and with the following steps:
\begin{enumerate}[label=\alph*.]
	\item We use the law of total probability, i.e., $\text{Pr}(X) = \sum_n \text{Pr}(Y_n) \cdot \text{Pr}(X|Y_n)$. Moreover, if two nodes $i$ and $j$ are not physical neighbors ($A_{ij}=0$), they cannot share any entangled links due to the absence of swaps, i.e., $\text{Pr}\big[w_{ij}(t)>0 \;|\;A_{ij}=0 \big] = 0$. Given the topology, $A_{ij}$ is a binary variable with a fixed value, therefore, $\text{Pr}(A_{ij}=1) = A_{ij}$.
	\item In a homogeneous network with no swaps, $w_{ij}$ depends on $A_{ij}$ but is otherwise independent of the nodes $i$ and $j$. Hence, $\mathbb{E}\big[w_{ij}(t) \;|\; A_{ij}=1\big]$ does not depend on $j$. This can also be seen in (\ref{eq.limexpwij}).
	\item The physical node degree of node $i$ can be computed as $d_i = \sum_{j=1}^{n} A_{ij}$.
	\item We use (\ref{eq.limexpwij}).
\end{enumerate}

(\ref{eq.expected_vi_noswaps}) and (\ref{eq.expected_ki_noswaps}) can be used to study the performance of the protocol in the limit of large number of resources ($r\rightarrow\infty$).
When $p_\mathrm{gen}>p_\mathrm{cons}$ we find
\begin{equation}
	\lim_{r\rightarrow\infty} v_i = d_i,
	\;\;\;\; \mathrm{and} \;\;\;\;
	\lim_{r\rightarrow\infty} k_i = \lim_{r\rightarrow\infty} r d_i = \infty.
\end{equation}
This means that, when the generation rate exceeds the consumption rate, the virtual neighborhood size will eventually saturate and every node will share entanglement with every physical neighbor. In particular, the average number of entangled links will increase infinitely (for large but finite $r$, $k_i$ reaches a maximum value of $\sim r d_i$).
When $p_\mathrm{gen}<p_\mathrm{cons}$,
\begin{equation}
	\lim_{r\rightarrow\infty} v_i = d_i \frac{p_\mathrm{gen}(1-p_\mathrm{cons})}{p_\mathrm{cons}(1-p_\mathrm{gen})},
	\;\;\;\; \mathrm{and} \;\;\;\;
	\lim_{r\rightarrow\infty} k_i = d_i p_\mathrm{gen} \frac{1-p_\mathrm{cons}}{p_\mathrm{cons}-p_\mathrm{gen}}.
\end{equation}

In Figure \ref{fig.analytical_no_swaps_extended}, we plot the expected virtual neighborhood size and expected virtual degree for different combinations of parameters, focusing on the interplay between $p_\mathrm{gen}$ and $p_\mathrm{cons}$. Both quantities decrease with increasing consumption rate, as one would expect, and quickly drop to zero for $p_\mathrm{cons}>p_\mathrm{gen}$.

\begin{figure}[ht!]
\captionsetup[subfigure]{justification=centering}
     \centering
     \begin{subfigure}[b]{0.4\textwidth}
         \centering
         \includegraphics[width=\textwidth]{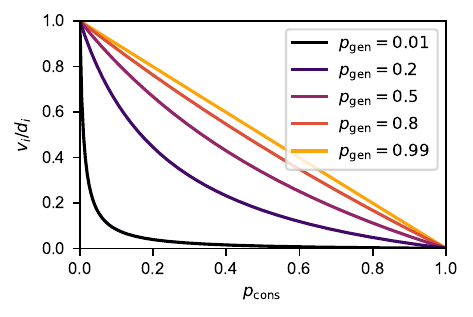}
	\vspace{-1.5\baselineskip}
         \caption{Virtual neighborhood size ($r=1$).}
         \vspace{5pt}
     \end{subfigure}
     \begin{subfigure}[b]{0.4\textwidth}
         \centering
         \includegraphics[width=\textwidth]{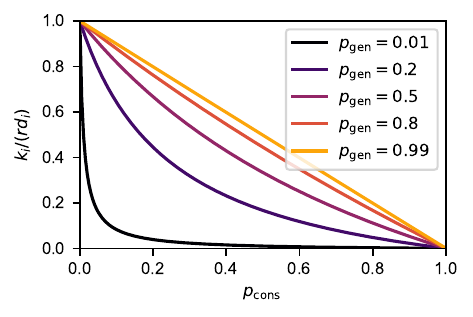}
	\vspace{-1.5\baselineskip}
         \caption{Virtual node degree ($r=1$).}
         \vspace{5pt}
     \end{subfigure}

     \begin{subfigure}[b]{0.4\textwidth}
              \centering
         \includegraphics[width=\textwidth]{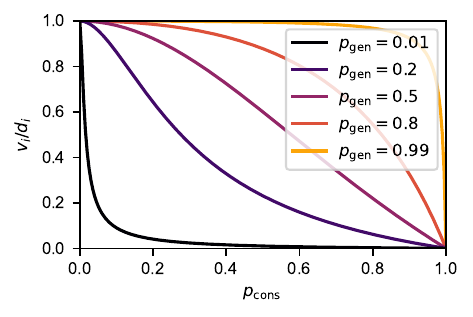}
	\vspace{-1.5\baselineskip}
         \caption{Virtual neighborhood size ($r=2$).}
         \vspace{5pt}
     \end{subfigure}
     \begin{subfigure}[b]{0.4\textwidth}
         \centering
         \includegraphics[width=\textwidth]{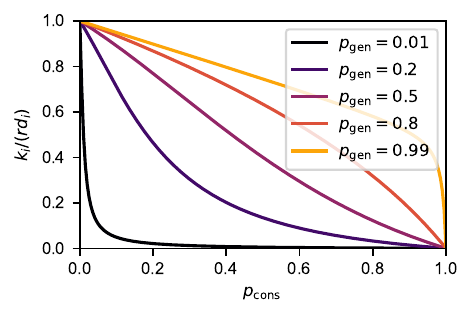}
	\vspace{-1.5\baselineskip}
         \caption{Virtual node degree ($r=2$).}
         \vspace{5pt}
     \end{subfigure}
     
     \begin{subfigure}[b]{0.4\textwidth}
              \centering
         \includegraphics[width=\textwidth]{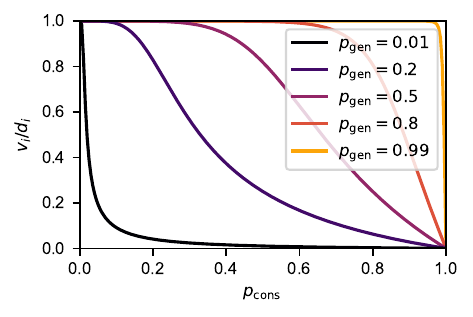}
	\vspace{-1.5\baselineskip}
         \caption{Virtual neighborhood size ($r=5$).}
         \vspace{5pt}
     \end{subfigure}
     \begin{subfigure}[b]{0.4\textwidth}
         \centering
         \includegraphics[width=\textwidth]{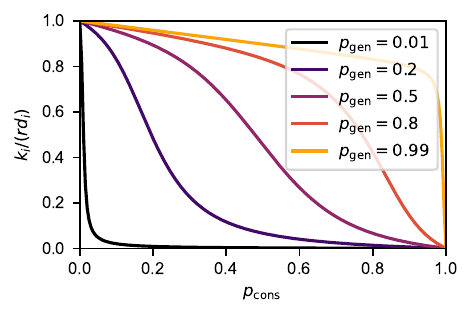}
	\vspace{-1.5\baselineskip}
         \caption{Virtual node degree ($r=5$).}
         \vspace{5pt}
     \end{subfigure}
     
     \begin{subfigure}[b]{0.4\textwidth}
              \centering
         \includegraphics[width=\textwidth]{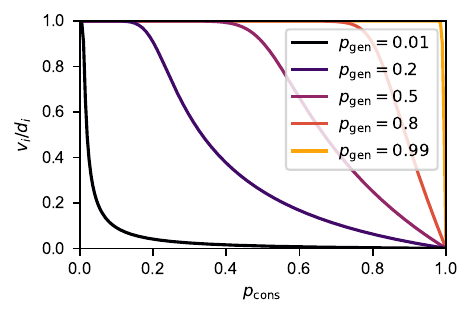}
	\vspace{-1.5\baselineskip}
         \caption{Virtual neighborhood size ($r=10$).}
         \vspace{5pt}
     \end{subfigure}
     \begin{subfigure}[b]{0.4\textwidth}
         \centering
         \includegraphics[width=\textwidth]{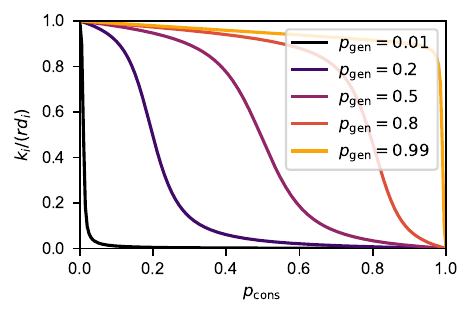}
	\vspace{-1.5\baselineskip}
         \caption{Virtual node degree ($r=10$).}
         \vspace{5pt}
     \end{subfigure}
     
	\caption{\textbf{Larger consumption rates decrease the virtual neighborhood size and the virtual node degree.}
	Expected virtual neighborhood size (left) and virtual node degree (right) in the steady state in a quantum network with no swaps, cutoff $t_\mathrm{cut}=10/p_\mathrm{cons}$ time steps, and $r=1,2,5,10$ qubits per node per physical channel (from top to bottom). All curves were calculated using Equations (\ref{eq.expected_vi_noswaps}) and (\ref{eq.expected_ki_noswaps}).
	}
	\label{fig.analytical_no_swaps_extended}
\end{figure}

Lastly, Figure \ref{fig.convergence_to_analytical} shows an example of the convergence of $\mathbb{E}\big[v_i(t) \big ]$ and $\mathbb{E}\big[k_i(t) \big ]$ to $v_i$ and $k_i$ over time, respectively. The time-dependent quantities have been calculated using a simulation on a quantum network with a $(2,3)$-tree physical topology. The dashed lines correspond to the steady-state values in the absence of swaps predicted by (\ref{eq.expected_vi_noswaps}) and (\ref{eq.expected_ki_noswaps}).

\begin{figure}[t!]
\captionsetup[subfigure]{justification=centering}
     \centering
     \begin{subfigure}[b]{0.4\textwidth}
         \centering
         \includegraphics[width=\textwidth]{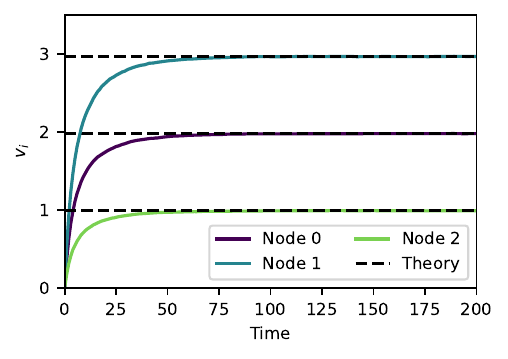}
	\vspace{-1.5\baselineskip}
         \caption{Expected virtual neighborhood size.}
         \vspace{5pt}
     \end{subfigure}
     \begin{subfigure}[b]{0.4\textwidth}
         \centering
         \includegraphics[width=\textwidth]{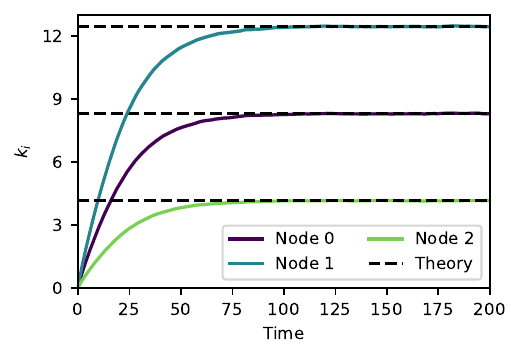}
	\vspace{-1.5\baselineskip}
         \caption{Expected virtual node degree.}
         \vspace{5pt}
     \end{subfigure}

	\caption{\textbf{The expected virtual neighborhood size and the expected virtual node degree converge to the steady-state analytical values in the absence of swaps.}
	In this example, we ran the SRS protocol (Algorithm \ref{alg.srs} from the main text) with $q=0$ (i.e., no swaps) on quantum network with $(2,3)$-tree topology.
	Nodes 0, 1, and 2 correspond to nodes in levels 0, 1, and 2 of the tree, respectively (i.e., they have physical node degrees $d_0=2$, $d_1=3$, and $d_2=1$, respectively).
	Each solid line corresponds to each of the three nodes.
	The dashed lines correspond to the expected steady-state values predicted by Equations (\ref{eq.expected_vi_noswaps}) and (\ref{eq.expected_ki_noswaps}) for each of the nodes.
	The standard error for each solid line is shown as a shaded region, although it is hard to notice since its maximum value is 0.017 in \textbf{(a)} and 0.056 in \textbf{(b)}.
	Other parameters used in this experiment are $p_\mathrm{gen}=0.2$, $F_\mathrm{new} = 0.9$, $r=5$, $T = 2000$ time steps, $p_\mathrm{cons}=0.1$, $F_\mathrm{app}=0.6$, $t_\mathrm{cut} = 221$ time steps (given by (\ref{eq.cutoffcondition})).
	Numerical results obtained using a network simulation and Monte Carlo sampling with $10^4$ samples.
	}
	\label{fig.convergence_to_analytical}
\end{figure}
\vspace{20pt}

\clearpage
\section{Steady state of a stochastic process}\label{app.numerical-steady-state}
In this Appendix we provide an algorithm to find the steady-state expected value of a stochastic process given a set of samples.
In our work, we employ this algorithm to estimate the steady-state expected value of the virtual neighborhood size, $\lim_{t\rightarrow\infty}\mathbb{E}\big[ v_i(t) \big]$, and the virtual node degree, $\lim_{t\rightarrow\infty}\mathbb{E}\big[ k_i(t) \big]$, from numerical simulations.

Finding the steady state of a stochastic process using realizations of the process is not a trivial task.
Algorithm \ref{alg.steadystate} can be used to estimate the start of the steady state of a stochastic process given $N$ realizations of the process.
The algorithm ensures that the expected values of the process at any two times in the steady state are arbitrarily close with a large probability.
We provide formal definitions and a proof below.

\begin{algorithm}
\caption{- Steady state estimation.}\label{alg.steadystate}
\begin{flushleft}
	\vspace{10pt}
	\textbf{Inputs:} \\
	\begin{itemize}[label=-]
		\item $\overline X_{\scriptscriptstyle N}(t)$, $t = t_0, t_1, ..., t_{\scriptscriptstyle M-1}$: sample mean of a stochastic process $X(t)$ over $N$ realizations at $t = t_0, t_1, t_2, ..., t_{\scriptscriptstyle M-1}$.
		\item $a$: minimum value of the stochastic process $X(t)$.
		\item $b$: maximum value of the stochastic process $X(t)$.
		\item $w$: minimum size of the steady state window.
	\end{itemize}
	\vspace{10pt}
	\textbf{Outputs:}
	\begin{itemize}[label=-]
		\item $\alpha$: the steady state is assumed to start at $t = t_\alpha$. The protocol aborts if it is not possible to find an $\alpha$ such that $\alpha \leq M-w$.
	\end{itemize}
	\vspace{10pt}
	\textbf{Algorithm:}
\end{flushleft}
\begin{algorithmic}[1]
	\State Define the error as $\varepsilon \leftarrow \frac{b-a}{\sqrt{N}}$.
	\State Define the steady state window: $W \leftarrow \{M-w, M-w+1, M-w+2, ..., M-1\}$. 
	\State Calculate $\Delta_{ij} \leftarrow 2\varepsilon - | \overline X_{\scriptscriptstyle N}(t_i) - \overline X_{\scriptscriptstyle N}(t_j) |$, $\forall i,j\in W$ and $i\neq j$.
	\State If $\Delta_{ij} < \frac{3}{2}\varepsilon$ for any $i,j$, then \textbf{abort} (steady state not found).
	\For{z in [1, 2, ..., M-w]}
		\State $k \leftarrow M-w-z$.
		\State Calculate $\Delta_{ik} \leftarrow 2\varepsilon - | \overline X_{\scriptscriptstyle N}(t_i) - \overline X_{\scriptscriptstyle N}(t_k) |$, $\forall i\in W$.
		\State If $\Delta_{ik} < \frac{3}{2}\varepsilon$ for any $i$, then $\alpha \leftarrow k+1$ and go to step 12.
		\State $W \leftarrow W \cup \{k\}$
	\EndFor
	\State $\alpha \leftarrow k$.
	\State \textbf{return} $\alpha$.
\end{algorithmic}
\end{algorithm}

\begin{theorem}
	Let $X(t) \in [a, b]$, with $a,b\in\mathbb{R}$, be a stochastic process with constant steady-state mean, i.e., $\lim_{t\rightarrow\infty} \mathbb{E}\big[X(t)\big] = X_\infty < \infty$.
	Let $\overline X_{\scriptscriptstyle N}(t_k)$ be a sample mean over $N$ samples at time $t_k \in \{t_0, t_1, ..., t_{\scriptscriptstyle M-1} \}$, with $t_0 < t_1 < ... < t_{\scriptscriptstyle M-1}$. Consider a minimum size of the steady-state window $w$.
	When $N\rightarrow\infty$, Algorithm \ref{alg.steadystate} with inputs $\overline X_{\scriptscriptstyle N}(t_k)$, $a$, $b$, and $w$, finds $\alpha$ such that
	$$\mathrm{Pr}\Big[\; \mathbb{E}\big[X(t_i)\big] \in \mathrm{IC}_{ij} \;\Big] \geq 0.815, \; \forall i,j\geq\alpha$$
	for an interval of confidence $\mathrm{IC}_{ij} = \Big(\; \mathrm{max}\big(\overline X_{\scriptscriptstyle N}(t_i), \overline X_{\scriptscriptstyle N}(t_j)\big) - \varepsilon, \; \mathrm{min}\big(\overline X_{\scriptscriptstyle N}(t_i), \overline X_{\scriptscriptstyle N}(t_j)\big) + \varepsilon \;\Big)$, with $\varepsilon = \frac{b-a}{\sqrt{N}}$,
	or the algorithm aborts.
\end{theorem}

\begin{proof}
Let us consider a stochastic process $X(t) \in [a, b]$ with constant steady-state mean, i.e., $\lim_{t\rightarrow\infty} \mathbb{E}\big[X(t)\big] = X_\infty < \infty$, and with finite variance $\sigma(t)^2$.
Assume that we have $N$ realizations of the process where we took samples at times $t_0 < t_1 < t_2, \dots$.
We denote the value taken in realization $n\in \{0,1,...,N-1\}$ at time $t$ as $x_n(t)$.
We define the sample average as
\begin{equation}
	\overline X_{\scriptscriptstyle N}(t) = \frac{1}{N}\sum_{n=0}^{N-1} x_n(t).
\end{equation}

The Central Limit Theorem states that the distribution of the random variable $\sqrt{N} \big( \overline X_{\scriptscriptstyle N}(t) - \mathbb{E}\big[X(t)\big] \big)$ converges to a normal distribution $\mathcal{N}(0,\sigma(t)^2)$ as $N$ approaches infinity.
After rescaling and shifting this distribution, we find that $\mathbb{E}\big[X(t)\big]$ converges to a normal distribution $\mathcal{N}\big(\overline X_{\scriptscriptstyle N}(t), \sigma(t)^2/N\big)$ as $N$ approaches infinity.
By the properties of the normal distribution,
\begin{equation}\label{eq.normal95}
	\mathrm{Pr}\Bigg[ \mathbb{E}\big[X(t)\big] \in \bigg( \overline X_{\scriptscriptstyle N}(t) - \frac{2\sigma(t)}{\sqrt{N}}, \; \overline X_{\scriptscriptstyle N}(t) + \frac{2\sigma(t)}{\sqrt{N}} \bigg) \Bigg] > 0.95.
\end{equation}
The values of $X(t)$ are constrained to the interval $[a,b]$, and therefore the standard deviation is upper bounded by \cite{Bhatia2000}
\begin{equation}\label{eq.boundstd}
	\sigma(t) \leq (b-a)/2.
\end{equation}
Let us define the error as $\varepsilon = \frac{b-a}{\sqrt{N}}$, and the interval of confidence for the expected value of $X(t_i)$ as
\begin{equation}\label{eq.IC_i}
	\mathrm{IC}_i = \Big( \overline X_{\scriptscriptstyle N}(t_i) - \varepsilon, \; \overline X_{\scriptscriptstyle N}(t_i) + \varepsilon \Big)
\end{equation}

Using (\ref{eq.normal95}), (\ref{eq.boundstd}), and (\ref{eq.IC_i}), we can write
\begin{equation}
	\mathrm{Pr}\Big[\mathbb{E}\big[X(t_i)\big] \in \mathrm{IC}_i  \Big] > 0.95.
\end{equation}

This result means that the expected value is arbitrarily close to the sample mean with high probability. Next, we need to show that any two expected values in the time window defined by the algorithm are arbitrarily close to each other to conclude that the window captures the steady-state behavior.

Let us define the interval of confidence $ij$ as the overlap in the intervals of confidence for the expected values of $X(t_i)$ and $X(t_j)$:
\begin{equation}\label{eq.IC_ij}
	\mathrm{IC}_{ij} = \Big(\; \mathrm{max}\big(\overline X_{\scriptscriptstyle N}(t_i), \overline X_{\scriptscriptstyle N}(t_j)\big) - \varepsilon, \; \mathrm{min}\big(\overline X_{\scriptscriptstyle N}(t_i), \overline X_{\scriptscriptstyle N}(t_j)\big) + \varepsilon \;\Big).
\end{equation}
The size of this interval of confidence is
\begin{equation}\label{eq.Delta_ij}
	\Delta_{ij} = 2\varepsilon - \big|\overline X_{\scriptscriptstyle N}(t_i) - \overline X_{\scriptscriptstyle N}(t_j)\big|.
\end{equation}
We provide a graphical intuition in Figure \ref{fig.intervals-of-confidence}.

\begin{figure}[h]
\captionsetup[subfigure]{}
\centering
  \centering
  \includegraphics[width=0.5\linewidth]{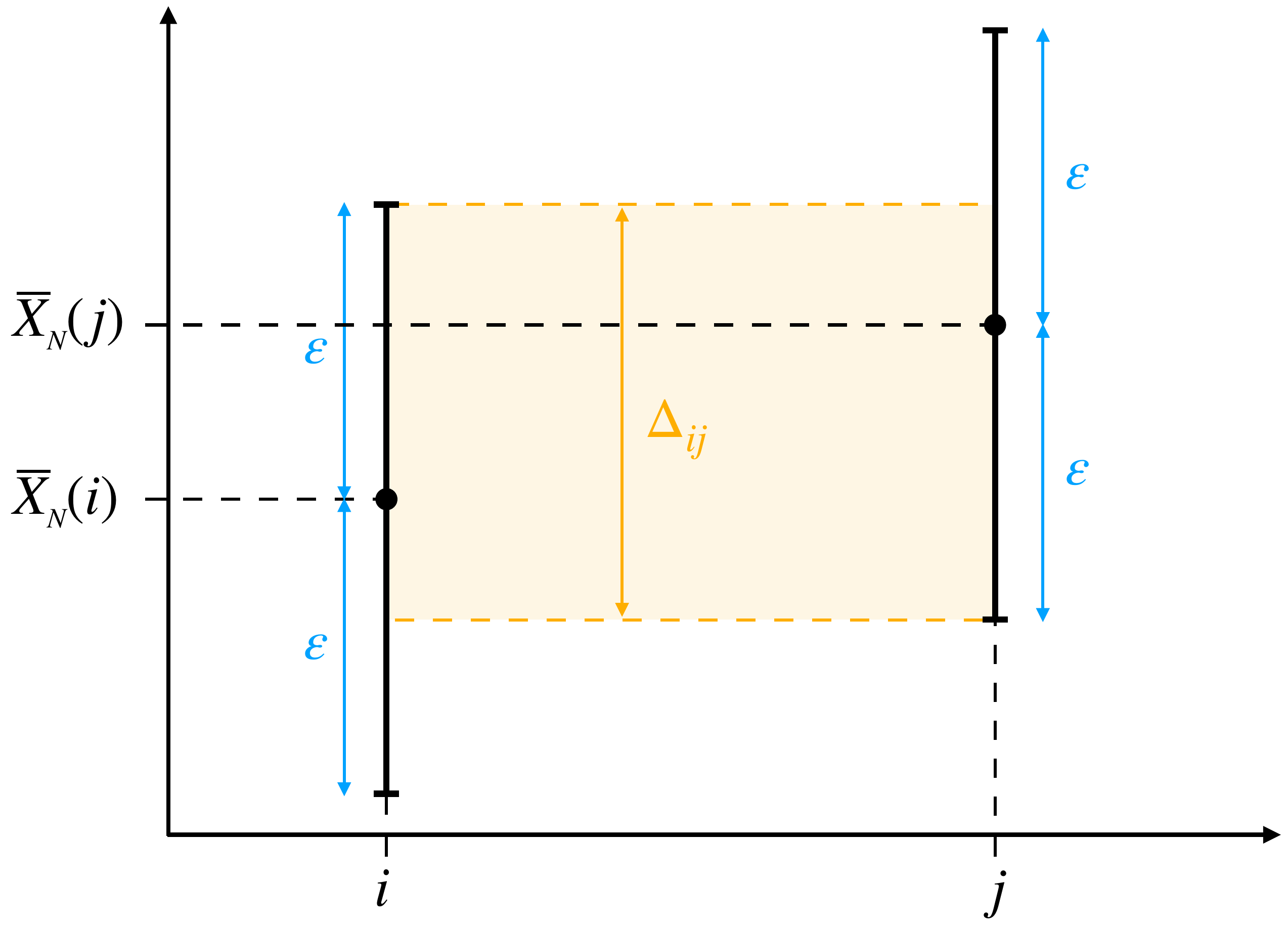}
\caption{\textbf{Graphical intuition for the interval of confidence $ij$ used to identify the steady state.}
	$\overline X_{\scriptscriptstyle N}(i)$ corresponds to the sample mean at $i$, $\varepsilon$ is the error, and $\Delta_{ij}$ is the size of the interval of confidence $ij$ (highlighted in yellow).}
\label{fig.intervals-of-confidence}
\end{figure}

Algorithm \ref{alg.steadystate} finds the smallest $\alpha$ such that $\alpha \leq M-w$ and $\Delta_{ij}\geq \frac{3}{2}\varepsilon$, for any $i,j>\alpha$. Then, we say that the steady state starts at $t_\alpha$. If $\alpha$ does not exist, the algorithm aborts.
Next, we show that the condition stated in the theorem,
\begin{equation}
	\mathrm{Pr}\Big[\; \mathbb{E}\big[X(t_i)\big] \in \mathrm{IC}_{ij} \;\Big] \geq 0.815, \; \forall i,j\geq\alpha,
\end{equation}
is equivalent to $\Delta_{ij}\geq \frac{3}{2}\varepsilon$, for any $i,j>\alpha$.
We proceed as follows:
\begin{equation}
\begin{split}
	\mathrm{Pr}\Big[\; \mathbb{E}\big[X(t_i)\big] \in \mathrm{IC}_{ij} \;\Big]
	&\stackrel{a}{=} \mathrm{Pr}\Big[\; \mathbb{E}\big[X(t_i)\big] \in \Big(\; \overline X_{\scriptscriptstyle N}(t_i) - \varepsilon, \; \overline X_{\scriptscriptstyle N}(t_i) + \Delta_{ij} - \varepsilon \;\Big) \;\Big]\\
	&\stackrel{b}{=} \int_{\overline X_{\scriptscriptstyle N}(t_i) - \varepsilon}^{\overline X_{\scriptscriptstyle N}(t_i)} f_{\scriptscriptstyle N}(x_i) \mathrm{d}x_i + \int_{\overline X_{\scriptscriptstyle N}(t_i)}^{\overline X_{\scriptscriptstyle N}(t_i) + \Delta_{ij} - \varepsilon} f_{\scriptscriptstyle N}(x_i) \mathrm{d}x_i \\
	&\stackrel{c}{\geq} \int_{\overline X_{\scriptscriptstyle N}(t_i) - \frac{2\sigma(t)}{\sqrt{N}}}^{\overline X_{\scriptscriptstyle N}(t_i)} f_{\scriptscriptstyle N}(x_i) \mathrm{d}x_i + \int_{\overline X_{\scriptscriptstyle N}(t_i)}^{\overline X_{\scriptscriptstyle N}(t_i) + \Delta_{ij} - \varepsilon} f_{\scriptscriptstyle N}(x_i) \mathrm{d}x_i \\
	&\stackrel{d}{\geq} \frac{0.95}{2} + \int_{\overline X_{\scriptscriptstyle N}(t_i)}^{\overline X_{\scriptscriptstyle N}(t_i) + \Delta_{ij} - \varepsilon} f_{\scriptscriptstyle N}(x_i) \mathrm{d}x_i \\
	&\stackrel{e}{\geq} 0.475 + \int_{\overline X_{\scriptscriptstyle N}(t_i)}^{\overline X_{\scriptscriptstyle N}(t_i) + \frac{\varepsilon}{2}} f_{\scriptscriptstyle N}(x_i) \mathrm{d}x_i \\
	&\stackrel{f}{\geq} 0.475 + \int_{\overline X_{\scriptscriptstyle N}(t_i)}^{\overline X_{\scriptscriptstyle N}(t_i) + \frac{\sigma}{\sqrt{N}}} f_{\scriptscriptstyle N}(x_i) \mathrm{d}x_i \\
	&\stackrel{g}{\geq} 0.475 +\frac{0.68}{2}\\
	&= 0.815
\end{split}
\end{equation}
with the following steps:
\begin{enumerate}[label=\alph*.]
	\item Without loss of generality, assume $\overline X_{\scriptscriptstyle N}(t_i) \geq \overline X_{\scriptscriptstyle N}(t_j)$.
	\item Let $f_{\scriptscriptstyle N}(x_i)$ be the probability distribution function of $\mathbb{E}\big[X(t_i)\big]$. As previously shown, when $N$ goes to infinity, this distribution converges to a normal distribution $\mathcal{N}\big(\overline X_{\scriptscriptstyle N}(t_i), \sigma(t)^2/N\big)$. We assume $N$ is sufficiently large.
	\item Using (\ref{eq.boundstd}): $\varepsilon = \frac{b-a}{\sqrt{N}} \geq \frac{2\sigma(t)}{\sqrt{N}}$.
	\item The probability that a normally distributed random variable takes a value between the mean and two standard deviations away is larger than $\frac{0.95}{2}$, i.e., $\int_{\mu-2\sigma}^{\mu} f(z) dz = \int_{\mu}^{\mu+2\sigma} f(z) dz \geq \frac{0.95}{2}$, where $f(z)$ is the probability distribution of $Z\sim\mathcal{N}(\mu,\sigma^2)$.
	\item The algorithm only considers $i$ and $j$ such that $\Delta_{ij}\geq\frac{3}{2}\varepsilon$.
	\item Using (\ref{eq.boundstd}) again: $\varepsilon \geq \frac{2\sigma(t)}{\sqrt{N}}$.
	\item The probability that a normally distributed random variable takes a value between the mean and one standard deviation away is larger than $\frac{0.68}{2}$, i.e., $\int_{\mu-\sigma}^{\mu} f(z) dz = \int_{\mu}^{\mu+\sigma} f(z) dz \geq \frac{0.68}{2}$, where $f(z)$ is the probability distribution of $Z\sim\mathcal{N}(\mu,\sigma^2)$.
\end{enumerate}
\end{proof}

Note that the validity of this method depends on the number of samples $N$, which must be sufficiently large in order to apply the Central Limit Theorem.

In our simulations, we employ Algorithm \ref{alg.steadystate} to check the existence of the steady state in the virtual neighborhood size, $v_i(t)$, and the virtual node degree, $k_i(t)$, of every node $i$.
After identifying the steady state, we take the average at the final simulation time as an estimate for the expected steady-state value, i.e., $\lim_{t\rightarrow\infty}\mathbb{E}\big[ v_i(t) \big] \approx \overline v_{i,{\scriptscriptstyle N}}(t_{\scriptscriptstyle M-1})$ and $\lim_{t\rightarrow\infty}\mathbb{E}\big[ k_i(t) \big] \approx \overline k_{i,{\scriptscriptstyle N}}(t_{\scriptscriptstyle M-1})$, where $\overline v_{i,{\scriptscriptstyle N}}(t)$ and $\overline k_{i,{\scriptscriptstyle N}}(t)$ are the sample averages at time $t$.
The virtual neighborhood size of node $i$ is upper bounded by $b=\min(rd_i, n)$, where $rd_i$ is the total number of qubits at node $i$ and $n$ is the total number of nodes.
The virtual degree of node $i$ is upper bounded by $b=rd_i$.
In this work, each simulation was run over $10 t_\mathrm{cut}$ time steps, and the window used to estimate the steady state was $w = 2 t_\mathrm{cut}$.

When the standard error is very small and the mean value is slowly converging to the steady-state value, the overlaps between intervals of confidence ($\Delta_{ij}$) may be too small.
Then, our algorithm may abort, indicating that there is not steady state.
In practice, we would like the algorithm to declare that the steady state has been reached once we are close enough to the steady-state value.
To prevent the algorithm from aborting in such a situation, we can increase the value of $b$ to increase the size of the interval of confidence ($\varepsilon$) in the algorithm.

We considered employing other data analysis techniques, such as bootstrapping and data blocking \cite{Thijssen2007}, to improve our estimates.
However, we decided to not use them since ($i$) bootstrapping would require running the simulations over many more time steps to be able to take many samples spaced an autocorrelation time; and ($ii$) data blocking requires a much larger storage space.

As a final remark, we measure the error in the estimate of the expected steady-state values using the standard error $\epsilon = s_{\scriptscriptstyle N}/\sqrt{N}$, where $s_{\scriptscriptstyle N}$ is the sample standard deviation.
In particular, the error bars used in this work correspond to $\pm 2 \epsilon$, which provide a $95\%$ interval of confidence.
\\

Figure \ref{fig.example-steady-state-finder} shows an example of our algorithm finding the steady state of the virtual neighborhood size when running the SRS protocol in a network with a tree topology.
The virtual neighborhood size of three nodes is shown in different colors. Dots correspond to the time $t_\alpha$ at which the algorithm declares that the steady state has been reached.
\\

\begin{figure}[h]
\captionsetup[subfigure]{}
\centering
  \centering
  \includegraphics[width=0.7\linewidth]{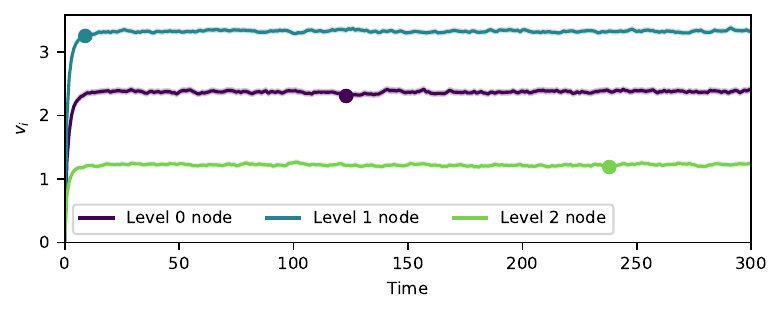}
\caption{\textbf{Algorithm \ref{alg.steadystate} can identify the steady state of a stochastic process.}
	Evolution of the average virtual neighborhood size in a quantum network with a $(2,3)$-tree topology running the SRS protocol described in the main text.
	Each line (purple, blue, and green) corresponds to a node in a different level of the tree (level 0, 1, and 2).
	Dots indicate that the steady state has been reached, according to Algorithm \ref{alg.steadystate}.
	The error for each solid line is shown as a shaded region, although it is hard to notice since its maximum value is 0.029 (the error is defined as $2\hat\sigma/N_\mathrm{samples}$, where $\hat\sigma$ is the sample standard deviation and $N_\mathrm{samples}$ is the number of samples).
	Other parameters used in this experiment: $p_\mathrm{gen}=0.9$, $F_\mathrm{new} = 0.88$, $p_\mathrm{swap}=1$, $r=5$, $T = 2000$ time steps, $M=4$, $p_\mathrm{cons}=0.225$, $q=0.1$, $F_\mathrm{app}=0.6$, $t_\mathrm{cut} = 56$ time steps.
	Numerical results obtained using a network simulation and Monte Carlo sampling with $10^3$ samples.
	The simulation was run over 560 time steps (only the first 300 are shown here) and the steady-state window was 112 time steps.}
\label{fig.example-steady-state-finder}
\end{figure}
\vspace{20pt}

\clearpage
\section{Extra experiments on a tree network}\label{app.extra}
Here, we provide more examples of the dependence of the virtual neighborhood size, $v_i$, and the virtual node degree, $k_i$, on the SRS protocol parameter $q$ (probability that a node performs a swap).
In the main text, we discuss the dependence on $q$ using a network with the following baseline set of parameters: $(2,3)$-tree topology, $p_\mathrm{gen}=0.9$, $F_\mathrm{new} = 0.888$, $p_\mathrm{swap} = 1$, $r=5$, $T = 2000$ time steps, $M=4$, $p_\mathrm{cons}=0.225$, $F_\mathrm{app} = 0.6$, $t_\mathrm{cut} = 56$ time steps.
Figure \ref{fig.extra_qdependence_trees} shows similar plots for networks with slightly different combinations of parameters that correspond to larger trees, smaller consumption rate, and probabilistic swapping.
In all situations we observe the same qualitative behavior as in the baseline case: the value of $q$ that maximizes the virtual neighborhood size is node-dependent, and $k_i$ is monotonically decreasing with increasing $q$.

\begin{figure}[ht!]
\captionsetup[subfigure]{justification=centering}
     \centering
     \begin{subfigure}[b]{0.4\textwidth}
         \centering
         \includegraphics[width=0.75\textwidth]{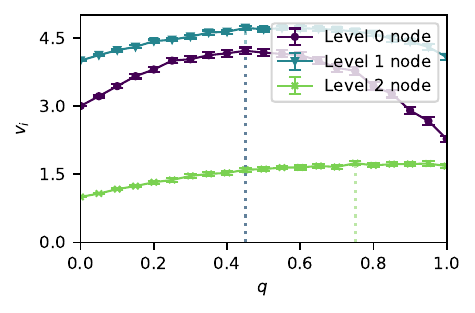}
	\vspace{-0.8\baselineskip}
         \caption{Virtual neighborhood size ($(3,3)$-tree).}
         \vspace{5pt}
     \end{subfigure}
     \hspace{40pt}
     \begin{subfigure}[b]{0.4\textwidth}
         \centering
         \includegraphics[width=0.75\textwidth]{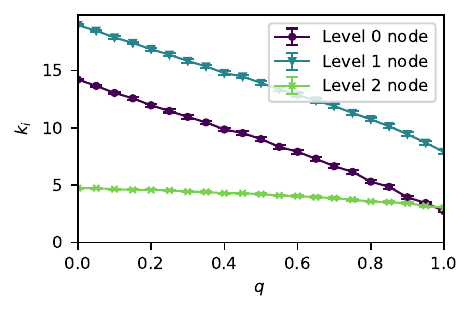}
	\vspace{-0.8\baselineskip}
         \caption{Virtual node degree ($(3,3)$-tree).}
         \vspace{5pt}
     \end{subfigure}

     \begin{subfigure}[b]{0.4\textwidth}
              \centering
         \includegraphics[width=0.75\textwidth]{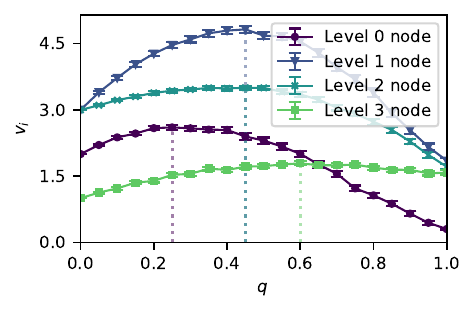}
	\vspace{-0.8\baselineskip}
         \caption{Virtual neighborhood size ($(2,4)$-tree).}
         \vspace{5pt}
     \end{subfigure}
     \hspace{40pt}
     \begin{subfigure}[b]{0.4\textwidth}
         \centering
         \includegraphics[width=0.75\textwidth]{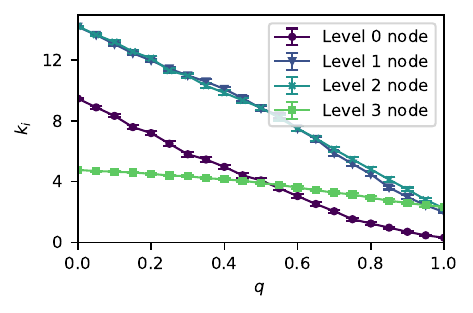}
	\vspace{-0.8\baselineskip}
         \caption{Virtual node degree ($(2,4)$-tree).}
         \vspace{5pt}
     \end{subfigure}
     
     \begin{subfigure}[b]{0.4\textwidth}
              \centering
         \includegraphics[width=0.75\textwidth]{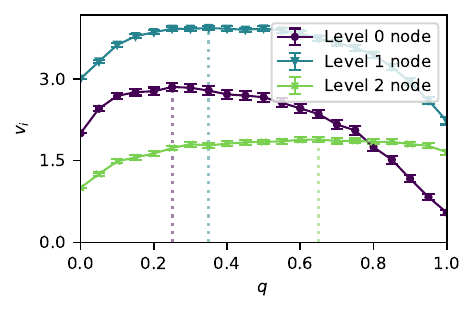}
	\vspace{-0.8\baselineskip}
         \caption{Virtual neighborhood size ($p_\mathrm{cons}=0.1$).}
         \vspace{5pt}
     \end{subfigure}
     \hspace{40pt}
     \begin{subfigure}[b]{0.4\textwidth}
         \centering
         \includegraphics[width=0.75\textwidth]{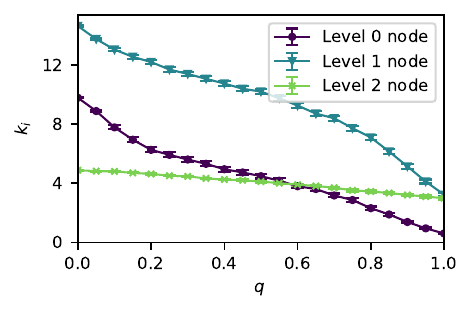}
	\vspace{-0.8\baselineskip}
         \caption{Virtual node degree ($p_\mathrm{cons}=0.1$).}
         \vspace{5pt}
     \end{subfigure}
     
     \begin{subfigure}[b]{0.4\textwidth}
              \centering
         \includegraphics[width=0.75\textwidth]{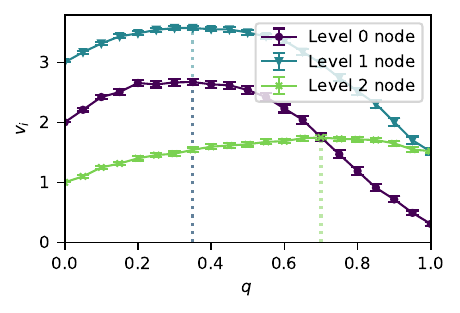}
	\vspace{-0.8\baselineskip}
         \caption{Virtual neighborhood size ($p_\mathrm{cons}=0.1$, $p_\mathrm{swap}=0.5$).}
         \vspace{5pt}
     \end{subfigure}
     \hspace{40pt}
     \begin{subfigure}[b]{0.4\textwidth}
         \centering
         \includegraphics[width=0.75\textwidth]{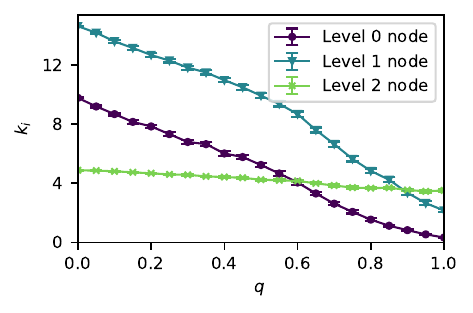}
	\vspace{-0.8\baselineskip}
         \caption{Virtual node degree ($p_\mathrm{cons}=0.1$, $p_\mathrm{swap}=0.5$).}
     \end{subfigure}
     
	\caption{\textbf{Our performance metrics show the same qualitative behavior for different combinations of parameters.}
	Expected virtual neighborhood size \textbf{(a, c, e, g)} and virtual node degree \textbf{(b, d, f, h)} in the steady state in a tree network running the SRS protocol vs the protocol parameter $q$.
	The value of $q$ that maximizes the virtual neighborhood size is indicated by the dotted lines.
	Baseline parameters:
	$(2,3)$-tree topology, $p_\mathrm{gen}=0.9$, $F_\mathrm{new} = 0.888$, $p_\mathrm{swap} = 1$, $r=5$, $T = 2000$ time steps, $M=4$, $p_\mathrm{cons}=0.225$, $F_\mathrm{app} = 0.6$, $t_\mathrm{cut} = 56$ time steps.
	 The subfigures in each row correspond to a different experiment (each caption indicates the parameters that have a different value in that experiment).
	Results obtained using a network simulation and Monte Carlo sampling with $10^3$ samples. The error in the error bars is defined as $2\hat\sigma/N_\mathrm{samples}$, where $\hat\sigma$ is the sample standard deviation and $N_\mathrm{samples}$ is the number of samples.
	}
	\label{fig.extra_qdependence_trees}
\end{figure}
\vspace{20pt}



\end{document}